\documentclass[10pt]{article}
\usepackage[final]{epsfig}
\usepackage{amsfonts, amsmath, amssymb, amsthm,bbm}
\usepackage{latexsym}
\usepackage{color}
\usepackage{changes}
\usepackage{verbatim}

\numberwithin{equation}{section}

\newtheorem{theorem}{Theorem}[section]
\newtheorem{lemma}[theorem]{Lemma}
\newtheorem{remark}[theorem]{Remark}
\newtheorem{example}[theorem]{Example}
\newtheorem{proposition}[theorem]{Proposition}
\newtheorem{definition}[theorem]{Definition}
\newtheorem{corollary}[theorem]{Corollary}

\newcommand{\R}{\mathbb{R}}

\newcommand{\C}{\mathbb{C}}
\newcommand{\Z}{\mathbb{Z}}
\newcommand{\N}{\mathbb{N}}

\newcommand{\be}{\begin{equation}}
\newcommand{\ee}{\end{equation}}

\newcommand{\cala}{{\mathcal{A}}}

\newcommand{\calP}{{\mathcal{P}}}

\newcommand{\bthe}{\begin{theorem}}
\newcommand{\ethe}{\end{theorem}}

\newcommand{\ben}{\begin{enumerate}}
\newcommand{\een}{\end{enumerate}}

\newcommand{\beq}{\begin{equation}}
\newcommand{\eeq}{\end{equation}}

\newcommand{\ble}{\begin{lemma}}
\newcommand{\ele}{\end{lemma}}

\newcommand{\bde}{\begin{definition}}
\newcommand{\ede}{\end{definition}}

\newcommand{\bco}{\begin{corollary}}
\newcommand{\eco}{\end{corollary}}

\newcommand{\bpr}{\begin{proposition}}
\newcommand{\epr}{\end{proposition}}

\newcommand{\bexam}{\begin{example}\rm}
\newcommand{\eexam}{\halmos\end{example}}

\newcommand{\beao}{\begin{eqnarray*}}
\newcommand{\eeao}{\end{eqnarray*}\noindent}

\newcommand{\beam}{\begin{eqnarray}}
\newcommand{\eeam}{\end{eqnarray}\noindent}

\newcommand{\barr}{\begin{array}}
\newcommand{\earr}{\end{array}}

\newcommand{\bproof}{\begin{proof}}
\newcommand{\eproof}{\end{proof}}

\newcommand{\eps}{\varepsilon}
\newcommand{\var}{{\rm var}}

\def\bbn{{\Bbb N}}

\newcommand{\halmos}{\quad\hfill\mbox{$\Box$}}

\newcommand{\CC}[1]{{\color{red} #1}}

\begin{document}

\begin{center}
{\LARGE  {Infinite-body optimal transport with Coulomb cost \\[2mm]}}  

\normalsize
\vspace{0.4in}

Codina Cotar$^1$, Gero Friesecke$^2$ and Brendan Pass$^3$  \\[1mm]
$^1\,$ Department of Statistical Science, University College London \\
$^2\,$Department of Mathematics, Technische Universit\"at M\"unchen  \\
$^3\,$ {Department of Mathematical and Statistical Sciences}, University of Alberta \\
c.cotar@ucl.ac.uk, gf@ma.tum.de, pass@ualberta.ca\\[2mm]
\end{center}

\noindent
\normalsize
{\bf Abstract.} We introduce and analyze symmetric infinite-body optimal transport (OT) problems with cost function of pair potential form.
We show that for a natural class of such costs, the optimizer is given by the independent product measure all of whose factors are given by the one-body marginal. This is in striking contrast to standard finite-body OT problems, in which the optimizers are typically highly correlated, as well as to infinite-body OT problems with Gangbo-Swiech cost. 
Moreover, by adapting a construction from the study of exchangeable processes in probability theory, we prove that the corresponding $N$-body OT problem is well approximated by the infinite-body problem. 

To our class belongs the Coulomb cost which arises in many-electron
quantum mechanics. The optimal cost of the Coulombic N-body OT problem as a function of the one-body marginal density is known in the physics and quantum chemistry literature under the name 
{\it SCE functional}, and arises naturally as the semiclassical limit of the celebrated Hohenberg-Kohn functional. Our results imply that in the inhomogeneous high-density limit (i.e. $N\to\infty$ with arbitrary fixed inhomogeneity profile $\rho/N$), the SCE functional converges to the mean field functional.  We also present reformulations of the infinite-body and N-body OT problems as two-body OT problems with representability constraints.
\vspace{0.2in}




\vspace{0.2in}

\section{Introduction}

{\bf Semi-classical electron-electron interaction functional and connection to optimal transport.} 
This work is motivated by, and contributes to, the longstanding quest in physics, chemistry and mathematics to design and justify approximations to the energy functional of many-electron quantum mechanics in terms of the one-body density. 

A simplified yet still formidable challenge consists in understanding the following ''semi-classical'' interaction energy functional obtained by a constrained search
over $N$-point densities with given one-body density $\rho$. 
This functional, introduced in the physics literature by Seidl, Perdew, Levy, Gori-Giorgi, and Savin \cite{Seidl99, SPL99, SGS07}, is given by 
\be \label{FOTN}
   V_{ee}^{SCE}[\rho] := \inf_{\gamma_N\in\calP^N_{sym}(\R^{3}),\,\gamma_N\mapsto\rho/N} C_N[\gamma_N], 
\ee
where $\rho$ is a given nonnegative function on $\R^3$ with $\int_{\R^3} \rho = N$ (physically: the total
electron density of an atom or molecule with $N$ electrons) and
\be \label{CN}
   C_N[\gamma_N] := \int_{\R^{3N}} 
   \sum_{1\le i<j\le N} \frac{1}{|x_i-x_j|} d\gamma_N(x_1,\ldots,x_N).
\ee
Here $\calP_{sym}^{N}(\R^3)$ is the space of probability measures $\gamma_N$ on $\R^{3N}$ which satisfy the symmetry condition
\be \label{symmetric}
    \gamma_N(A_1\times\cdots\times A_N) = \gamma_N(A_{\sigma(1)}\times\cdots\times A_{\sigma(N)})
    \mbox{ for all Borel sets}~A_1,\ldots,A_N\subseteq\R^3 \mbox{ and all permutations }\sigma,
\ee
and the notation $\gamma_N\mapsto\rho/N$ means that $\gamma_N$ has one-body density $\rho$ (physics terminology) or equivalently equal
$\R^3$-marginals $\rho/N$ (probability terminology),
\be \label{marginals}
    \gamma_N(\R^{3(i-1)}\times A_i \times \R^{3(N-(i-1))}) = \int_{A_i}\frac{\rho(x)}{N}\,dx \mbox{ for all }A_i\subseteq\R^3
    \mbox{ and all }i=1,\ldots,N.
\ee
The normalization factor $1/N$ in (\ref{FOTN}) and (\ref{marginals}) is owed to the convention in many-electron quantum mechanics 
that the one-body density $\rho$ should integrate to the number of particles in the system, i.e. $\int_{\R^3} \rho = N$, whereas
the marginal density in the sense of probability theory, denoted in the sequel by $\mu$, should integrate to $1$. The functional (\ref{FOTN}) is commonly called the {\it SCE functional}, where the acronym SCE stands for strictly correlated electrons; the fact that e.g. for $N=2$, minimizers concentrate on lower-dimensional sets of form $x_2=T(x_1)$ (see (\ref{Monge}) below) has the physical interpretation that given the position of the first electron, the position of the second electron is strictly determined. The connection of the functional (\ref{FOTN}) with many-electron quantum mechanics which motivated this work is explained at the end of this Introduction. 
\\[2mm]
We remark that dropping the symmetry requirement on $\gamma_N$ would not alter the minimum value in (\ref{FOTN}), since the functional $C_N$ takes the same value on a nonsymmetric measure as on its symmetrization.
\\[2mm]
Because of the appearance of the $N$-particle
configurations $(x_1,\ldots,x_N)$ and of the $N$-body cost $\sum_{i<j}1/{|x_i-x_j|}$ in {$C_N[\gamma_N]$}, we call this
functional an $N$-{\it body mass transportation functional} or {\it an optimal transport problem with $N$ marginals}, and the problem (\ref{FOTN}) of minimizing it an $N$-{\it body optimal transport problem}. The functional $V_{ee}^{SCE,N}$ can be interpreted as the {\it minimum cost of an optimal transport problem as a functional of the marginal measure}. In the case $N=2$, one is dealing with a standard (two-body or two-marginal) optimal transport problem of form
$$
    \mbox{Minimize }\int_{\R^{2d}} c(x_1,x_2) d\gamma_2(x_1,x_2) \mbox{ over } \gamma_2\in\calP(\R^{2d})~ \\
    \mbox{subject to }\gamma_2(A\times\R^3)=\gamma_2(\R^3\times A)=\mu(A) \mbox{ for all }A\subseteq\R^3,
$$
where $c\, : \, \R^d\times\R^d\to\R\cup\{\infty\}$ is a {\it cost function} and $\calP(\R^{2d})$ is the space of probability measures on $\R^{2d}$. 
\\[2mm]
{\bf Previous results} 
It was not realized until recently \cite{CFK11, BPG12} that the minimization problem in (\ref{FOTN}) has the form of an optimal transport problem and
can, especially in the case $N=2$, be fruitfully analyzed via methods from OT theory.

OT problems with two marginals have been studied extensively in the mathematical literature for a large variety of cost functions; see, for example \cite{Bre87}, and \cite{GM96} for some influential results in the area and \cite{Vill09} for a comprehensive treatment. A central insight in this setting is that, under fairly weak conditions on the cost function and marginals, the optimal measure is unique and of Monge type, i.e. it concentrates on the graph of a map over $x_1$. That is to say,
\be \label{Monge}
   \gamma_2=(I\times T)_\sharp \mu \mbox{ (OT notation) or equivalently }\gamma_2(x,y)=\mu(x)\delta_{T(x)}(y) \mbox{ (physics notation) for some map }
   T\, : \, \R^d\to\R^d.
\ee
Even though the Coulomb cost lies outside the costs treated in standard OT theory (where positive power costs like $|x-y|$ or $|x-y|^2$ are prototypical),
the result (\ref{Monge}) has recently been extended to the 2-body OT problem with Coulomb cost, (\ref{FOTN}) with $N=2$ \cite{CFK11, BPG12}, confirming
earlier nonrigorous results in the physics literature \cite{Seidl99, SGS07}.  
\\[2mm]
Much less is known about $N$-body OT problems with $N\ge 3$. Here the OT literature has focused on special cost functions
\cite{Rus91}, \cite{RusUck97}, \cite{GS98},  \cite{Hein02}, \cite{CarNaz08}, \cite{Pass10}, \cite{Pass11}, \cite{CFK2},  \cite{Pass12}, \cite{Pass12c}, \cite{BPG12}, \cite{CFKMP13}, \cite{CD13}, \cite{CDD13}, \cite{GhouMoa13}, \cite{KimPass13} and the structure of solutions is highly dependent on the cost function. For certain costs, solutions concentrate on graphs over the first marginal, as in the two body case, while for others the solutions can concentrate on high dimensional submanifolds of the product space. In particular, despite its importance in electronic structure theory, very little is known regarding the structure of the solutions of the $N$-body OT problem with Coulomb cost (\ref{FOTN}).  Let us note, however, that the study of Monge-Kantorovich problems with symmetry constraints has been intitiated in \cite{GhouMoa13a} and continued in \cite{GhouMoa13}, \cite{GhouMaur13}, \cite{GalGhou13}, \cite{CDD13} and \cite{CD13}, the last two papers dealing with the Coulomb cost.
\\[2mm]
{\bf Main results} 
Here we focus on problem (\ref{FOTN}) in the {\it regime of large $N$}, i.e. the ''opposite'' regime of the hitherto best understood case $N=2$. We present two main results. The first introduces and analyzes the associated 
{\it infinite-body OT problem}. Remarkably, for a natural class of costs which includes the Coulomb cost, the infinite-body problem is uniquely minimized by the 
independent product measure all of whose factors are given by the one-body marginal. See Theorem 1.1 below for the precise statement. This stands in surprising contrast to the pair of recent papers [Pass12a] 
and [Pass12b]. There costs of Gangbo-Swiech type are analyzed and it is shown that the optimizer is a Monge type solution; that is, any two of the variables are completely dependent rather than completely independent. Our second main result says that the corresponding $N$-body OT problem is well approximated by the infinite-body problem; in particular we show that the optimal cost per particle pair of the $N$-body problem converges to that of the infinite-body problem
as $N$ gets large. See Theorem 1.2 for the precise statement.
\\[2mm]
{\bf Connection with many-electron quantum mechanics and the Hohenberg-Kohn functional}
Next let us explain the connection with, and implications for, many-electron quantum mechanics.
Heuristically, the functional $V_{ee}^{SCE}$ is the semiclassical limit of the celebrated Hohenberg-Kohn functional \cite{HK64},
\be \label{semiclass}
     V_{ee}^{SCE}[\rho] = \lim_{\hbar\to 0} F^{HK}[\rho], 
\ee
where
\be \label{FHK}
     F^{HK}[\rho]:=\min_{\Psi\in\cala_N, \, \Psi\mapsto\rho} \langle \Psi , ( \hbar^2 \widehat{T} +  \widehat{V}_{ee})\Psi\rangle.  
\ee
Here $\widehat{T}=-\frac12\Delta$, $\Delta$ is the Laplacian on $\R^{3N}$, and the resulting contribution to the functional is
the quantum mechanical kinetic energy of the system, $\widehat{V}_{ee}$ is the electron-electron operator which acts
by multiplication with the function $V_{ee}(x_1,..,x_N)=\sum_{1\le i<j\le N}1/|x_i-x_j|$, ${\cala_N}$ denotes the set of
antisymmetric, square-integrable functions $\Psi \, : \, (\R^3\times\Z_2)^N\to\C$ with square-integrable gradient 
and $L^2$ norm $1$, $\langle \cdot,\cdot\rangle$ is the $L^2$ inner product, and the notation $\Psi\mapsto\rho$ means that the associated $N$-point position density 
\beq \label{density}
   \gamma_N(x_1,..,x_N) = \sum_{s_1,..,s_N\in\Z_2} |\Psi(x_1,s_1,..,s_N,x_N)|^2
\eeq
satisfies $\gamma_N\mapsto\rho/N$. The class of single-particle densities on which $F^{HK}$ is defined is the image of ${\cal A}_N$ under the map $\Psi\mapsto\rho$. By a result of Lieb [Li83], 
this class equals the set of functions $\rho \, : \, \R^3\to\R$ which are nonnegative, have integral $N$, and have the property that $\sqrt{\rho}$ belongs to the Sobolev space $H^1(\R^3)$. The HK functional constituted the birth of modern density functional theory (DFT). DFT approximates 
$F^{HK}$ by simpler yet still remarkably accurate functionals of the one-body density amenable to efficient numerical minimization, and is the currently most widely used method for numerical electronic structure computations for complex systems ranging from condensed matter over surfaces and nanoclusters to large
molecules. For further information about the HK functional and mathematical aspects of the challenge to approximate it by computationally simpler functionals we refer to our recent paper \cite{CFK11} and the literature cited therein.
A rigorous justification of eq. (\ref{FHK}) is given in \cite{CFK11} (for N=2) and \cite{CFK2} (for an arbitrary number of particles). While the proof itself shall not concern us here, we remark that there is indeed something to prove:
minimizers $\gamma_N$ of the limit problem in (\ref{FOTN}) are typically singular measures and hence do not arise as $N$-point densities (\ref{density}) of any quantum wavefunction $\Psi\in\cala_N$, making it a nontrivial task to construct a wavefunction with precisely the same one-body density as $\gamma_N$ for which the quantum expectation value on the right hand side of (\ref{FHK}) is well defined and close to the value $V_{ee}^{SCE}[\rho]=C_N[\gamma_N]$ on the left hand
side of (\ref{semiclass}). 
\\[2mm]
Together with the companion result (\ref{semiclass}), Theorem 1.2 says that the Hohenberg-Kohn functional $F^{HK}$ is rigorously asymptotic, in the regime of small $\hbar$, a large number of electrons, and a fixed inhomogeneity profile $\rho/N$, to the mean field functional 
\be \label{J}
   J[\rho] = \frac12 \int_{\R^6} \frac{1}{|x_1-x_2|}\rho(x_1)\rho(x_2)\, dx_1\, dx_2.
\ee
See Corollary \ref{IHDL} below for the precise statement. 
This result answers an open question raised by us in \cite{CFKMP13}, where we observed this correspondence for a toy model (one-body densities supported on two points, cost favouring different-site occupancy over same-site occupancy) for which the $N$-body OT problem in (\ref{FOTN}) can be solved explicitly. 
\\[2mm]
{\bf Precise statement of main results}
With a view to the application to density functional theory, we will work in the following setting even though some of our main results could be stated and proved for more general spaces, such as Polish spaces for 
Theorem \ref{convergence}.

Let $(\Omega_i^d, {\cal {F}}_i^d):=(\R^d, {\cal B}(\R^d))$, where $i=1,2,\ldots, N,\ldots,$ and $d\ge 1$. The underlying $\sigma$-field is the
Borel $\sigma$-field. Let $(\Omega_\infty^d, {\cal B}_\infty^d)$ be defined as the cartesian product of  $(\Omega_i^d, {\cal F}_i^d$, that is, $\Omega_\infty^d:=\prod_{i=1}^\infty\Omega_i^d$, and ${\cal B}_\infty^d$ is the Borel $\sigma$-algebra generated by the open subsets of $\Omega_\infty^d$ of the form $\prod_{i=1}^\infty A_i$, $A_i\in{\cal B}(\R^d)$, where $A_i=\Omega_i^d$ for all but a finite number of $i$. To simplify the notation, we will write $(\R^d)^\infty$ instead of $\Omega_\infty^d$. Throughout the paper, if $\mu\in\calP(\R^d)$ has a Lebesgue-integrable density, the latter is also denoted by $\mu$.

For all $N\in\N, N\ge 2$, let the \textit{cost function} $c_N:\underbrace{\R^d\times\ldots\times\R^d}_\textrm{N times} \rightarrow {\R}_+\cup\{\infty\}$ be defined by
\begin{equation}
\label{generalcost}
    c_N(x_1,\ldots,x_N):=\sum_{1\le i<j\le N} c(x_i,x_j),
\end{equation}
where $c:\R^{2d}\rightarrow [0,\infty)\cup\{\infty\}$ is assumed throughout to be Borel-measurable and symmetric (the latter means that $c(x,y)=c(y,x)$ for all $x,y\in\R^d$). For any $N\in\N$, and any infinite-dimensional probability measure $\gamma$ belonging to the space ${\calP}^{\infty}_{sym}(\R^d)$ defined below, let
\begin{equation}
\label{cnformula}
  {C}_N[\gamma] =  \int_{(\R^d)^\infty} c_N(x_1\ldots,,x_N)d\gamma(x_1,x_2,\ldots,x_N,\ldots)= \sum_{1\le i<j\le N} \int_{(\R^d)^\infty} c(x_i,x_j)d\gamma(x_1,x_2\ldots,x_N,\ldots).
\end{equation}
Here the domain of this functional is the space $\calP_{sym}^\infty(\R^d)$ of symmetric Borel probability measures on $(\R^d)^\infty$. For a more detailed discussion of the notion of infinite-dimensional symmetric Borel measures see for example \cite{DF80}. Symmetric means that for all $N$ and for all $N$-tupel 
$(i_1,..,i_N)$ of indices with $1\le i_1<i_2<...<i_N$, 
$$
    \gamma(\R^{d(i_1-1)}\times A_{i_1}\times\R^{d(i_2-i_1-1)}\times A_{i_2}\times\ldots\times A_{i_N}\times\R^d\times\cdots) = \gamma(\R^{d(i_1-1)}\times 
    A_{\sigma({i_1})}\times\R^{d(i_2-i_1-1)}\times A_{\sigma({i_2})}\times\ldots \times A_{\sigma(i_N)}\times\R^d\times\cdots) ,
$$ 
for all Borel sets $A_{i_1}, A_{i_2},\ldots A_{i_N}\subset\R^d$ and for all permutations $\sigma$ of $\{ i_1,i_2,\ldots, i_N\}$. 
As $N\to\infty$, the problem of minimizing $C_N$ subject to the marginal constraint $\gamma\to\mu$ turns into a meaningful, and -- as we shall see -- very interesting, limit problem:
\be \label{FtildeOTgen}
   \mbox{Minimize } C_\infty[\gamma] := \lim_{N\to\infty} \frac{1}{{N\choose 2}} C_N[\gamma] \mbox{ over infinite-dimensional probability measures }
   \gamma\in\calP^\infty_{sym}(\R^d) \mbox{ with }\gamma\to\mu.
\ee
Here the standard notation $\gamma\mapsto\mu$ means that $\gamma$ has one-body marginal $\mu$, i.e. 
$\gamma(A\times \prod_{i=1}^\infty \R^d) = \mu(A)$ for all Borel $A\subset\R^d$. A key object of interest is the optimal cost of the problem (\ref{FtildeOTgen}) as a function of the marginal measure, 
\be \label{FOTinfty}
     {F}^{OT}_\infty[\mu] = \inf_{\gamma\in\calP_{sym}^\infty(\R^d),\, \gamma\mapsto\mu} C_\infty[\gamma].
\ee
Because of the appearance of the infinite particle
configurations $(x_1,\ldots ,x_N,\ldots)$ and of an infinite-body cost, we call the problem (\ref{FtildeOTgen}) an {\it infinite-body (or infinite-marginal) optimal transport problem}. 
\\[2mm]
The large-$N$ limit of the DFT functional $V_{ee}^{SCE}$ described in the Introduction corresponds to the case $d=3$ and the \textit{Coulomb cost} $c(x,y)=\frac{1}{|x-y|}$. In this case, the functional (\ref{FOTinfty}) becomes 
{\be \label{FtildeOT}
     {F}^{OT}_\infty[\mu] := \inf_{\gamma\in\calP_{sym}^\infty(\R^3),\,\gamma\mapsto\mu} \lim_{N\rightarrow\infty}\frac{1}{{N\choose 2}}C_N[\gamma], \;\;\;
     {C}_N[\gamma] = \int_{(\R^3)^\infty} 
     \sum_{1\le i<j\le N} \frac{1}{|x_i-x_j|} d\gamma(x_1,\ldots,x_N,\ldots).
\ee 
Our first main result is the following. Here and below, $\hat{f}$ denotes the Fourier transform of the function $f\in L^1(\R^d)$, defined by $\hat{f}(k)=\int_{\R^d} e^{-ik\cdot x} f(x)\, dx$, and $C_b(\R^d)$ denotes the space of bounded continuous functions on $\R^d$. 
\begin{theorem} (Mean field theory as exact solution to infinite-body optimal transport)
\label{definetti1}
\begin{itemize}
\item[(a)] Let $c \, : \, \R^{2d}\to\R_+\cup\{\infty\}$  in (\ref{generalcost}) be of the form $c(x,y) =\ell(x-y)$, where $\ell(z)=\ell(-z)$ for all $z\in\R^d$ 
(i.e. $c$ is symmetric), and either \\[1mm]
(i) $\ell\in L^1(\R^d)\cap C_b(\R^d)$, $\hat{\ell}\ge 0$ or \\
(ii) $d=3$, $\ell(z)=1/|z|$ (Coulomb cost). \\[1mm]
Let $\mu\in{\cal P}(\R^d)$ be a measure such that 
\be \label{finitenesscdn}
                    \int_{\R^{2d}} c(x,y) \, \mu(dx)\mu(dy)<\infty.
\ee 
Then the independent measure 
\be \label{indep}
     \gamma_0=\mu^{\otimes\infty}=\mu\otimes\mu\otimes\cdots
\ee
is a minimizer of the infinite-body optimal transport problem (\ref{FtildeOTgen}), and the optimal
cost is the mean field functional, i.e. 
\begin{equation}
  \label{infOTexplic}
  {F}^{OT}_{\infty}[\mu]=\int_{\R^{2d}} c(x,y) \, \mu(dx)\mu(dy).
\end{equation}
\item[{(b)}] If in addition $\hat{\ell}(z)$ is strictly bigger than zero for all $z$, then the
independent measure (\ref{indep}) is the unique minimizer of the problem (\ref{FtildeOTgen}). 
\end{itemize}
\end{theorem}
Note that in case (ii), i.e. the Coulomb cost in dimension $d=3$, the strict positivity condition
$\hat{\ell}>0$ holds, because $\hat{\ell}(k)=4\pi/|k|^2$. Moreover by simple estimates (see e.g. eq. (5.21) in the proof of Theorem 5.6 in \cite{CFK11}) the finiteness condition in (a) holds for all $\mu\in L^1(\R^3)\cap L^3(\R^3)$; the latter is the natural $L^p$ type space into which the domain of the Hohenberg-Kohn functional embeds. As a consequence, the above results are valid for all densities of physical interest in DFT.  
However the Coulomb cost it is neither continuous nor does it belong to $L^1$. The obvious task to weaken the regularity assumptions in (i) so as to naturally include the Coulomb cost does not seem to be straightforward and lies beyond the scope of this article.
\\[2mm]
Our result stands in surprising contrast to the recent results in \cite{Pass12a, Pass12b} by one of us. For a class of of costs including the many-body quadratic cost $\sum_{i \neq j}|x_i-x_j|^2$ studied by Gangbo-Swiech \cite{GS98}, the optimizer of the infinite-body OT problem is demonstrated to be a Monge type solution; that is, any two of the variables are completely dependent, rather than completely independent as is the case for our class of costs.  This dichotomy exposes a fascinating sensitivity to the cost function in infinite-body optimal transport problems. This difference is not present in two-marginal problems, where fairly weak conditions on the cost which include both the quadratic and the Coulomb cost suffice to ensure Monge type solutions. 
A milder version of the dichotomy does however arise in the multi-body context, where for certain costs the solution can concentrate on high dimensional submanifolds of the product space \cite{CarNaz08}, \cite{Pass10}.  It does not seem to be until one gets to the infinite marginal setting, however, that complete independence of the variables becomes optimal for certain costs. The difference between the costs in our paper and those in \cite{Pass12a, Pass12b} can be expressed succinctly as positivity of the Fourier transform of $\ell$. Note that the latter is equivalent to the fact that $c(x,y)=\ell(x-y)$ is a {\it positive kernel}, i.e. associated integral operator $K\varphi(x):=\int_{\R^d}c(x,y)\varphi(y)\, dy$ satisfies $\langle \varphi, K\varphi\rangle\ge 0$ for all $\varphi\in C_0^\infty(\R^d)$. See Example \ref{Ex:nonposFT} (ii) in Section \ref{2} for a simple explicit example of a cost function which satisfies all the assumptions in Theorem \ref{definetti1} except positivity of the Fourier transform and for which the conclusion of the theorem fails.

The basic idea for the proof of Theorem \ref{definetti1} is to represent the competing infinite-dimensional probability measures in (\ref{FOTinfty}) via de Finetti's theorem, and identify the functional $C_\infty$ introduced in (\ref{FtildeOTgen}), with the help of Fourier transform calculus and elementary probability theory, as a sum of the mean field functional and a certain variance term minimized by completely independent measures.

Our second main result clarifies the relationship between the infinite-body optimal transport problem (\ref{FtildeOTgen}) 
and the corresponding $N$-body optimal transportation problem: 
\be \label{Nbodyproblem}
   \mbox{ Minimize }\tilde{C}_N[\gamma_N] :=  \int_{\R^{dN}} \sum_{1\le i<j\le N}c(x_i,x_j)\, d\gamma_N(x_1,x_2\ldots,x_N) 
   \mbox{ over }\gamma_N\in\calP^N_{sym}  (\R^d) \mbox{ satisfying }\gamma_N\mapsto\mu.
\ee
Here and below $\calP^N_{sym}(\R^d)$ denotes the set of Borel probability measures $\gamma_N$ on $\R^{Nd}$ which are symmetric, i.e. satisfy eq. (\ref{symmetric}) (with $\R^3$ replaced by $\R^d$). 
The optimal cost per particle pair as a function of the marginal measure will be denoted by 
$F^{OT}_N[\mu]$; that is to say, for arbitrary $\mu\in\calP(\R^{d})$ we set
\begin{equation}
\label{finOT0}
    {F}^{OT}_N[\mu] :=  \frac{1}{{N\choose 2}} \inf_{\gamma_N\in\calP_{sym}(\R^{dN}),\gamma_N\,\mapsto\mu}{\tilde{C}}_N[\gamma_N].
\end{equation}
We show:
\begin{theorem} (N-body cost approaches infinite-body cost)
\label{convergence}
Assume that $\mu\in\calP(\R^d)$ is a probability measure such that there exists a measure $\gamma_0\in \calP_{sym}^\infty(\R^{d})$ with $\gamma\mapsto\mu$ and $\int_{(\R^d)^\infty} c(x_1,x_2)d\gamma_0(x_1,x_2,\ldots)<\infty$. Let the cost function $c \, : \, \R^{2d}\to[0,\infty)\cup \{+\infty\}$ in 
(\ref{Nbodyproblem}) and (\ref{cnformula}) be Borel-measurable, symmetric, and either (i) bounded, or (ii) lower semi-continuous as a map with values into $[0,\infty)\cup\{+\infty\}$ endowed with its natural topology; ie, $c(x_j) \rightarrow\infty$ whenever $x_j \rightarrow x$ and $c(x) = \infty$.Then we have
\begin{equation}
\label{limitOT}
    F^{OT}_\infty[\mu]=\lim_{N\rightarrow\infty} {F}^{OT}_N[\mu].
\end{equation}
\end{theorem}
Note that here not just costs leading to independence as in Theorem \ref{definetti1} but also costs leading to strong correlations as considered in \cite{Pass12a, Pass12b} are included.

The proof of Theorem \ref{convergence} is based on a construction from advanced probability theory \cite{DF80} which does not appear to be easily accessible to non-probabilists, and which contains the important insight that any $N$-body measure $\gamma_N\in\calP^N_{sym}(\R^{d})$ can be approximated by the $N$-body marginal $\tilde{\gamma}_N$ of an infinite probability measure $\gamma\in\calP_{sym}^\infty(\R^d)$ ($\tilde{\gamma}_N$ is infinitely representable in the terminology developed below). This allows us to approximate the $N$-body OT problem (\ref{Nbodyproblem}) as arising in density functional theory by the corresponding infinite-body OT problem (\ref{FtildeOTgen}). Interestingly, the focus of probabilists was precisely the other way around: the object of primary interest were the infinite probablity measures in the space $\calP_{sym}^\infty$, or in fact the underlying infinite sequences of random variables. The latter serve as useful alternatives to iid (identically and independently distributed) sequences which allow to model repeated sampling experiments containing correlations; approximation by finite sequences of random variables was then of interest for purposes of numerical sampling. 
\\[2mm]
Finally let us describe what our results imply for the SCE functional (\ref{FOTN}), (\ref{CN}) arising in density functional theory. Roughly, they allow to analyze a natural {\it inhomogeneous high-density limit} in which the inhomogeneity is not a small perturbation, but stays proportional to the overall density. More precisely, one fixes an arbitrary density $\mu$ of integral $1$, considers the $N$-body system with proportional inhomogeneity, i.e. with one-body density given by $\rho=N\mu$,
and studies the asymptotics of the SCE energy as $N$ gets large. Note that the SCE energy corresponds, up to normalization factors, to the optimal cost functional (\ref{finOT0}) with Coulomb cost $c(x,y)=1/|x-y|$ in dimension $d=3$:
\begin{equation} \label{Veescaled}
     V_{ee}^{SCE}[\rho] = {N \choose 2} F_N^{OT}[\frac{\rho}{N}].
\end{equation}
Combining Theorem \ref{definetti1} and Theorem \ref{convergence} immediately yields:
\begin{corollary} \label{IHDL} (Inhomogeneous high-density limit of the SCE functional) Let $\mu\, : \, \R^3\to\R$ be any nonnegative function with $\int_{\R^3}\mu=1$ which belongs to $L^1(\R^3)\cap L^3(\R^3)$. Let $\rho^{(N)}=N\mu$. Then as $N$ gets large, the SCE energy of $\rho^{(N)}$ is
asymptotic to the mean field energy, that is to say
$$
    \lim_{N\to\infty} \frac{V_{ee}^{SCE}[\rho^{(N)}]}{J[\rho^{(N)}]} = 1,
$$  
where $J$ is the functional (\ref{J}).  
\end{corollary}

We remark that both numerator and denominator are of order $N^2$ as $N\to\infty$, i.e. they are proportional to the number of particle pairs in the system. A very interesting question raised by our work is to determine asymptotic corrections to the mean field energy in eq. (\ref{limitOT}). For non-singular costs, we expect the next-order correction to
occur at the thermodynamic order $O(N)$. Unfortunately, understanding these corrections lies beyond the scope of the methods developed here.
\begin{remark}\label{lieboxford}
A very interesting alternative proof of the preceding corollary for the Coulombic cost function was pointed out to us by Paola Gori-Giorgi. This proof, and hence also the above corollary, is implicit in recent work in the physics literature \cite{RSG11}. The key ingredient is a nontrivial Coulombic inequality, the Lieb-Oxford bound \cite{LO81},  
The argument is as follows: the Lieb-Oxford bound, in our notation, states that
\begin{equation*}
V_{ee}^{SCE}[\rho^{(N)}]-J[\rho^{(N)}] \geq -C \int_{\mathbb{R}^3} (\rho^{(N)})^{4/3},
\end{equation*}
for some constant $C$ independent of $N$. (Strictly speaking, the bound was only formulated and derived in \cite{LO81} for $N$-point densities which arise from some wavefunction, but the proof generalizes easily to probability measures.) Noting that the left hand side is non positive (by using the independent $N$-point density as trial function in the
variational principle for $V_{ee}^{SCE}$), and that $V_{ee}^{SCE}[\rho^{(N)}]$ and $J[\rho^{(N)}]$ scale like $N^2$ while $ \int_{\mathbb{R}^3} (\rho^{(N)})^{4/3} = N^{4/3}\int_{\mathbb{R}^3} \mu^{4/3}$ scales as $N^{4/3}$, we divide by $J[\rho^{(N)}]$ and let $N$ tend to $\infty$ to obtain the desired result.

The arguments developed in the present paper apply to a larger class of interaction energies (see Theorems 1.1, 1.2), and - perhaps more importantly - are based on a general and transparent probabilistic inequality (namely the comparison estimate in Proposition 3.2 below between infinitely representable and finitely representable measures which goes back to Diaconis and Freedman). But -- unlike the Lieb-Oxford inequality -- our arguments fail to give a quantitative error bound for the associated optimal cost functionals for singular costs like the Coulomb cost, yielding such bounds only in the case of bounded costs (see eq. (\ref{similartolieboxford})).   
\end{remark}
{\bf Plan of paper.} 
The rest of the paper is organized as follows. In section \ref{0} we recall the notion of $N$-representability of pair measures, which was developed in the present OT context in our recent paper \cite{CFKMP13} and is equivalent to the concept of $N$-extendability of pairs of random variables in probability theory, and prove Theorem \ref{definetti1}. Section \ref{4} is devoted to the proof of Theorem \ref{convergence}.

\section{Solution to the infinite-body OT problem}
\label{0}
The proof of Theorem \ref{definetti1} will require two key Lemmas.
The first one (Lemma \ref{reduction}) reduces the infinite-body OT problem (\ref{FtildeOTgen}) to a $2$-body OT problem with an infinite representability constraint. The second (Lemma \ref{reduction1}) gives an explicit description of the
measures satisfying this infinite representability constraint (de Finetti's Theorem, stated in Proposition \ref{definetti} below). 

In subsection \ref{1} we recall the notion of $N$-representability of a pair density, generalize it
to infinitely many particles, prove Lemmas \ref{reduction} and \ref{reduction1}, and
also establish existence of at least one solution to (\ref{FtildeOTgen})} (Proposition \ref{existinfcase}).
In subsection \ref{2} we establish Theorem \ref{definetti1}, via Fourier transform calculus applied to the de Finetti representation of infinitely representable measures. 

\subsection{Reduction to a $2$-body OT problem with infinite representability constraint}
\label{1}

We now reformulate the infinite-body mass transportation problem
(\ref{FtildeOTgen}) as a standard (two-body) mass transportation problem subject to an infinite representability
constraint. This reformulation is possible due to the fact that the cost in (\ref{generalcost}) is a sum of symmetric pair terms. We begin by recalling the definition of $N$-representability, introduced in the present context in our recent paper \cite{CFKMP13} (see Definition III.1). 
\begin{definition}(N-representability) \label{Def1} Let $N\ge 2$. A symmetric probability measure $\mu_2\in
\calP_{sym}(\R^{2d})$ is said to be \emph{$N$-representable} if there exists a symmetric probability measure $\gamma_N\in\calP_{sym}^N(\R^{d})$ such that for all Borel sets $A_i,A_j\subseteq\R^d$ and all $1\le i<j\le N$, we have
\be \label{densityNrep}
    \gamma_N(\R^{d(i-1)}\times A_i \times \R^{d(j-(i-1))}\times A_j\times \R^{d(N-(j-1))}) = 
    \mu_2(A_i\times A_j). 
\ee
\end{definition}

$N$-representability is a highly nontrivial restriction. The following basic example is taken from \cite{CFKMP13}. 
\\[2mm]
{\bf Example} Let $A$, $B\in\R^d$, $A\neq B$. The totally anticorrelated probability measure 
$\mu_2 = \frac12(\delta_A\otimes\delta_B + \delta_B\otimes\delta_A)$ is not 3-representable. (Here $\delta_A$ denotes the Dirac measure centred at $A$.) 
\\[2mm]
Intuitively, this is because we can not allocate 3 particles to 2 sites without doubly occupying one of the sites. Mathematically, to prove this suppose that $\gamma$ was any probability measure on $(\R^d)^3$ with two-body marginal $\mu_2$. Then $\gamma$ must have one-body marginal supported on $\{A,B\}$, and hence must be a convex combination of the measures $\delta_X\otimes\delta_Y\otimes\delta_Z$ with $X,Y,Z\in\{A,B\}$. But the two-point marginal of each of the latter measures contains a positive multiple of either $\delta_A\otimes\delta_A$ or $\delta_B\otimes\delta_B$, whence the two-pont marginal of $\gamma$ cannot equal $\mu_2$. For further discussion and more general examples we refer to \cite{CFKMP13}. 
\\[2mm]
Two quantum analogues of $N$-representability are widely studied in the physics and quantum chemistry literature. The first one, (wavefunction) representability of a pair density, is closely related
to the notion above and asks whether a symmetric nonnegative function $p_2\, : \, \R^{2d}\to\R$ of unit integral satisfies
$$
  p_2(x_1,x_2)=\sum_{s_1,..,s_N\in\Z_2} \int_{\R^{d(N-2)}}|\Psi(x_1,s_1,x_2,s_2,...,x_N,s_N)|^2
$$
for some square-integrable antisymmetric normalized $N$-electron wavefunction $\Psi\, : \, (\R^d\times\Z_2)^N\to\C$. Wavefunction representability trivially implies representability in the sense of the definition above. Conversely, many known necessary conditions on representability by an $N$-electron wavefunction, such as the Davidson \cite{Da95} and generalized Davidson \cite{AD06} constraints, continue to hold for pair densities which are $N$-representable in the sense of Definition \ref{Def1}, as their derivation in fact only uses representability by a symmetric probability measure. 

In the second quantum analogue, one asks whether a function $\Gamma\, : \, (\R^d\times\Z_q)^4 \to \C$ is of the form
\be \label{quantrep}
   \Gamma(z_1,z_2; z_1',z_2') = \int_{\R^{(N-2)d}} \Psi(z_1,z_2,z_3,..,z_N)\overline{\Psi(z_1',z_2',z_3,..,z_N)} \, dz_3 ... dz_N
\ee
for some antisymmetric function $\Psi\in L^2((\R^d\times \Z_q)^N)$ of unit norm, with the case of electrons corresponding to $d=3$, $q=2$. Mathematically, $\Gamma$ should be viewed as a unit-trace operator $\hat{\Gamma}$ on the two-body Hilbert space $L^2((\R^d\times\Z_q)^2)$, acting as 
$$
   \varphi \mapsto (\hat{\Gamma}\varphi)(z_1,z_2) = \int_{(\R^d\times\Z_q)^2} \Gamma(z_1,z_2;z_1',z_2') \varphi(z_1',z_2')\, dz_1' dz_2'.
$$
Eq. (\ref{quantrep}) means that $\hat{\Gamma}$ can be represented as a partial trace of the unit-trace
operator $|\Psi\rangle\langle\Psi|$ on the $N$-body Hilbert space $L^2((\R^d\times\Z_q)^N)$. For an overview of results on the quantum representability problem we refer to \cite{COLYUK00}. 

The notion of $N$-representability in Definition \ref{Def1} is well known in the probability theory literature, under the names {\it N-extendability} or {\it finite exchangeability}, and is usually stated and analyzed in the language of sequences $X_1,..,X_N$ of $N$ random variables. The formulation in Definition \ref{Def1} is mathematically equivalent and corresponds to considering instead the law of the
random vector $(X_1,..,X_N)$. Numerous attempts have been made to characterize $N$-extendability for $N\ge 3$ for various types of marginals (see, for example, \cite{Ald85} for an an in-depth overview of  $N$-extendability results in probability), but a direct characterization remains elusive.

Let us now generalize Definition \ref{Def1} to infinite particle systems. 

\begin{definition}(Infinite representability) Analogously to the $N$-representability case, a symmetric probability measure $\mu_2\in
\calP_{sym}(\R^{2d})$ is said to be \emph{infinitely representable} if there exists a symmetric probability measure
$\gamma_\infty\in\calP_{sym}^\infty(\R^d)$ such that for all Borel sets $A_i,A_j\subseteq\R^d$ and all $1\le i<j\le N$, we have
\be \label{densityrep}
    \gamma_\infty(\R^{d(i-1)}\times A_i \times \R^{d(j-(i-1))}\times A_j\times \R^d\times\ldots\times\R^d\times \ldots) = 
    \mu_2(A_i\times A_j).
\ee
\end{definition}

Note that a symmetric probability measure $\gamma_\infty\in\calP_{sym}^\infty(\R^d)$ is called an \emph{exchangeable} measure in the probabilistic literature.
It is easy to see (see, for example, \cite{Ald85} or Lemma III.2 in \cite{CFKMP13}) that

\begin{lemma}
\label{Nrepressmall}
Let $N\ge M \ge 2$. If $\mu_2\in\calP_{sym}(\R^{2d})$ is $N$-representable, then it is also $M$-representable.
\end{lemma}
That is to say, $N$-representability becomes a more and more stringent condition as $N$ increases. 

We will next reformulate the minimization problem (\ref{FtildeOTgen}) in {terms} of infinite representability. The result is a straightforward extension to infinite particle systems of Theorem III.3 in \cite{CFKMP13} for the $N$-body problem.  
\begin{lemma}
\label{reduction}
 For any $\mu\in\calP(\R^d)$ we have
 \begin{equation}
 \label{infFOT} 
 {F}^{OT}_{\infty}[\mu] = \inf\Bigl\{\int_{\R^{2d}} c(x,y) \, d\mu_2(x,y) \; 
   \Big| \; 
   \mu_2\in \calP_{sym}^2(\R^{d}), \, \mu_2\mapsto\mu, \, \mu_2 \mbox{ is infinitely representable} \Bigr\}.
\end{equation}
\end{lemma}
\begin{proof}
This is clear from the observation that ${C}_{\infty}[\gamma] = \int_{\Omega^2}cd\mu_2$ for any $\gamma \in \calP_{sym}^\infty(\R^d)$  with $\gamma \rightarrow \mu_2$ and definition of infinite representability.
\end{proof}
To prove our next result, we will use the de Finetti-Hewitt-Savage Theorem for infinitely representable measures as stated and proved, in the different but equivalent language of exchangeable sequences of random variables, e.g. in Theorems 14 and 20 from \cite{DF80}:
\begin{proposition} (De Finetti-Hewitt-Savage Theorem)
\label{definetti}
Let $S$ be $\R^d$, or more generally any Polish space and ${\cal B}$ is the Borel $\sigma$-field.
Let ${\cal P}(S)$ be the set of probability measures on $(S, \cal{B})$, and let ${\cal{B}}^*(S)$ be the Borel $\sigma$-field in ${\cal P}(S)$. Let $\gamma_\infty$ be a symmetric Borel measure on the Borel $\sigma$-field ${\cal{B}}^\infty(S)$ of the product $S^\infty$  (for more precise definitions of these sets, see \cite{DF80}). Then there exists a unique Borel probability measure $\nu$ on ${\cal{B}}^*(S)$ such that
\begin{equation}
\label{definettieqn}
\gamma_\infty=\int_{ {\cal P}(S)} Q^{\otimes \infty}d\nu(Q).
\end{equation}
\end{proposition}
In words: one can view an infinitely representable probability measure as an integral of product probability measures against a probability measure defined on the space of probability measures. 

Next we reformulate the optimal cost functional (\ref{FOTinfty}) with the help of de Finetti's theorem.
\begin{theorem}
\label{reduction1}
For any $\mu\in\calP(\R^d)$, the functional $F_\infty^{OT}[\mu]$ introduced in (\ref{FOTinfty}) satisfies
\begin{equation}
 \label{inf1FOT} 
 {F}^{OT}_{\infty}[\mu] = \inf\Bigl\{\int_{\R^{2d}} c(x,y) \, d\mu_2(x,y) \; 
   \Big| \; \mu_2=\int_{\calP(\R^d)}Q\otimes Q \, d\nu(Q) ~\mbox{ and } \mu = \int_{\calP(\R^d)}Q \, d\nu(Q) ~\mbox{ for some}~{\nu\in\cal P}({\cal P }(\R^d))\Bigr\}.
\end{equation}
Moreover $\gamma=\int_{\calP(\R^d)}Q^{\otimes\infty} d\nu(Q)$ is a minimizer of the problem (\ref{FtildeOTgen}) if and only if $\mu_2=\int_{\calP(\R^d)} Q\otimes Q \, d\nu(Q)$ is a minimizer of the problem in (\ref{inf1FOT}).  
\end{theorem}
\begin{proof}
Note that the one and two body marginals of $\gamma_{\infty}$ in \eqref{definettieqn} are given by $\mu = \int_{\calP(\R^d)} Q \, d\nu(Q)$ and $\mu_2 = \int_{\calP(\R^d)}Q\otimes Q \, d\nu(Q)$, respecitvely. Then, by de Finetti's Theorem, $\mu_2$ is infinitely representable if and only if $\mu_2 = \int_{\calP(\R^d)}Q\otimes Q \, d\nu(Q)$ for some $\nu \in  {\cal P}({\cal P}(\R ^d))$.  The result follows from Lemma \ref{reduction}.
\end{proof}
We end this subsection with a general result of existence of at least one solution to (\ref{FtildeOTgen}) and to (\ref{infFOT}). This result will be used in the proof of Theorem \ref{convergence}.
\begin{theorem}
\label{existinfcase}
For all $N\in\N, N\ge 2$, let $c_N\, : \, (\R^d)^N \rightarrow {\R}_+\cup\{\infty\}$  be defined as in (\ref{generalcost}), with $c$ Borel-measurable, symmetric, and lower semi-continuous.
Then there exists at least one solution $\gamma^{opt}$ to (\ref{FtildeOTgen}) and at least one solution $\mu_2^{opt}$ to the minimization problem in (\ref{infFOT}).
\end{theorem}
\begin{proof}
To prove the existence of a solution $\gamma\in\calP_{sym}^\infty(\R^d),\gamma\mapsto\mu$, to (\ref{FtildeOTgen}), we will adapt to our infinite-body optimal transportation problem the standard proof
of existence of solutions to two-body OT problems as given e.g. in \cite{Vill09}, Theorem 4.1.
Since there are some subtle differences to the proof in \cite{Vill09}, we will outline below the basic steps. 

The proof relies on basic variational arguments involving the topology of weak convergence
(imposed by bounded continuous test functions). There are two key properties on which the proof relies:
\begin{itemize}
\item [(a)] Lower semicontinuity of the cost functional $\gamma\mapsto C_\infty[\gamma]$ on $\calP_{sym}^\infty(\R^d)$ with respect to weak convergence.   This follows by a standard argument after rewriting  $C_\infty[\gamma] = \int_{\R^{2d}}c(x_1,x_2)d\mu_2(x_1,x_2)$ and by noting that the class of infinite-dimensional symmetric probability measures in $\calP_{sym}^\infty(\R^d)$ is closed under weak convergence (for a proof of this statement, see e.g. page 54 in \cite{Ald85}).

\item [(b)] Tightness in $\calP_{sym}^\infty(\R^d)$ of the set of all $\gamma\in\calP_{sym}^\infty(\R^{d})$ such that $\gamma\mapsto\mu$ for some fixed $\mu\in\calP(\R^{d})$.

This is proved similarly to Lemma 4.3 from \cite{Vill09}. More precisely, let $\gamma\in\calP_{sym}^\infty(\R^{d})$ such that $\gamma\mapsto\mu$ and $\mu\in\calP(\R^{d})$. Since $\R^d$ is a Polish space, $\mu$ is tight in $\calP(\R^d)$. Then for any $\epsilon>0$ and for any $i\in\N, i\ge 1$,
there exists a compact set $K^i_\epsilon\subset\R^d$, independent of the choice of $\mu$, such that $\mu(\R^d\setminus K^i_\epsilon) \le\frac{\epsilon}{2^i}$. Take $K_\epsilon:=\prod_{i\ge 1}K^i_\epsilon$, which is compact by Tychonoff's theorem. Then we have
$$\gamma(K_\epsilon^c)\le \gamma(\cup_{i\ge 1}(\underbrace{\R^d\times\ldots\times\R^d}_\textrm{i-1 times}\times (K_\epsilon^i)^c\times\R^d\times\ldots))\le\sum_{i\ge 1}\mu((K_\epsilon^i)^c)\le \sum_{i\ge 1}\frac{\epsilon}{2^i}=\epsilon.$$
Tightness now follows since this bound is independent of $\gamma$.

\end{itemize}

Given (a) and (b), the existence of a solution ${\gamma}^{opt}$ to (\ref{FtildeOTgen}) follows analogously to the proof of Theorem 4.1 from \cite{Vill09}: take a minimizing sequence $\gamma^\alpha$, 
extract a weakly convergent subsequence via (b) and Prokhorov's theorem, and pass to the limit via (a).  

One now trivially also obtains a solution to the variational problem in (\ref{infFOT}); namely,
the two-point marginal $\mu_2^{opt}$ of $\gamma^{opt}$ is a solution.
\end{proof}

\subsection{Proof of Theorem \ref{definetti1}}
\label{2}

In this subsection, we determine explicitly the optimal transport functional  ${F}^{OT}_\infty$
introduced in eq. (\ref{FtildeOTgen}), for a large class of cost functions.  As an offshot, we obtain
an interesting probabilistic interpretation of the infinite-body optimal transport functional $C_\infty$ introduced in (\ref{FtildeOTgen}). 

{\bf Proof of Theorem \ref{definetti1}}

We will show explicitly that
\begin{equation}\label{ineq}
\int_{\R^{2d}} c(x,y) \, d\mu_2(x,y) \geq \int_{\mathbb{R}^{2d}} c(x,y) d\mu(x)d\mu(y)
\end{equation}
for any $\mu_2 =\int Q\otimes Qd\nu(Q)$ with $\nu \in  {\cal P}({\cal P}(\R ^d))$, and, if $\hat l>0$ everywhere, equality can only hold when $\mu_2 =\mu \otimes \mu$ is product measure.  The result then follows easily from Theorem \ref{reduction1}.

The central idea  is to re-write both terms in \eqref{ineq}  using  Fourier calculus and elementary probability theory.
For any $Q\in\calP(\R^d)$ such that $\int_{\R^{2d}}\ell(x-y)dQ(x)dQ(y)<\infty$, let $\ell*Q$ and $\hat{Q}$
denote, respectively, the convolution of $\ell$ and $Q$ and the Fourier transform of $Q$, i.e.
$$(\ell*Q)(x ):=\int_{\R^d}\ell(x-y)dQ(y),~~\hat{Q}(z)=\int_{\R^d}e^{-i z\cdot x} dQ(x). $$
The first function may take the value $+\infty$, whereas the second is a bounded continuous function on $\R^d$. In order not to obscure the main argument, we first calculate the integral in (\ref{inf1FOT}) {\it formally}, using the rules of Fourier transform calculus even though $\ell$ and $Q$ are not smooth rapidly decaying functions. The calculation will be justified rigorously in Lemma \ref{Fourier} below. Using, in order of appearance, Fubini's theorem, the definition of the convolution, Plancherel's formula, the Fourier calculus rule $\widehat{f*g}=\hat{f} \, \hat{g}$, and again Fubini's theorem gives 

\begin{eqnarray}
\label{fourier1}
\int_{\mathbb{R}^{2d}} c(x,y) d\mu_2(x,y) &=& \int_{\mathbb{R}^{2d}} \ell(x-y)\int_{\calP(\mathbb{R}^d)}dQ(x)dQ(y)d\nu(Q)\nonumber\\
&=&  \int_{\calP(\mathbb{R}^d)}\int_{\mathbb{R}^{2d}}\ell(x-y)\, dQ(x)\, dQ(y)\, d\nu(Q)\nonumber\\
&=&  \int_{\calP(\mathbb{R}^d)}\int_{\mathbb{R}^{d}}(\ell*Q)(x)\, dQ(x)\, d\nu(Q)\nonumber\\
&=&  \int_{\calP(\mathbb{R}^d)}\int_{\mathbb{R}^{d}}(\widehat{\ell*Q})(z)\bar{\hat{Q}}(z)\, dz \, d\nu(Q)\nonumber\\
&=&  \int_{\calP(\mathbb{R}^d)}\int_{\mathbb{R}^{d}}\hat{\ell}(z)|\hat{Q}(z)|^2 dz\, d\nu(Q)\nonumber\\
&=& \int_{\mathbb{R}^{d}} \hat{\ell}(z) \int_{\calP(\mathbb{R}^d)}|\hat{Q}(z)|^2d\nu(Q)dz.
\end{eqnarray}
By a similar reasoning, we have
\begin{eqnarray}
\label{fourier2}
\int_{\R^{2d}} c(x,y) d\mu(x)d\mu(y) &=& \int_{\R^{2d}} \ell(x-y)\int_{\calP(\R^d)\times \calP(\R^d)}dQ(x)\, d\nu(Q)\, d\tilde{Q}(y) d\nu(\tilde{Q})\nonumber\\
&=&\int_{\calP(\mathbb{R}^d)\times \calP(\mathbb{R}^d)}\int_{\mathbb{R}^{2d}}\ell(x-y)dQ(x)d\tilde{Q}(y)d\nu(Q)d\nu(\tilde{Q})\nonumber\\
&=&\int_{\calP(\mathbb{R}^d)\times \calP(\mathbb{R}^d)} \int_{\mathbb{R}^{d}}(\ell*\tilde{Q})(x)dQ(x)d\nu(Q)d\nu(\tilde{Q})\nonumber\\
&=&\int_{\calP(\mathbb{R}^d)\times \calP(\mathbb{R}^d)}\int_{\mathbb{R}^{d}}(\widehat{\ell*\tilde{Q}})(z)\bar{\hat{Q}}(z)dz\, d\nu(Q)\, d\nu(\tilde{Q})\nonumber\\
&=&\int_{\calP(\mathbb{R}^d)\times \calP(\mathbb{R}^d)} \int_{\mathbb{R}^{d}}\hat{\ell}(z)\hat{\tilde{Q}}(z)\bar{\hat{Q}}(z)\, dz\, d\nu(Q)\, d\nu(\tilde{Q})\nonumber\\
&=&\int_{\mathbb{R}^{d}}\hat{\ell}(z)\Big|\int_{P(\mathbb{R}^d)}\hat{Q}(z)\, d\nu(Q)\Big|^2dz.
\end{eqnarray}
Finally, decomposing the expressions on the right hand side of (\ref{fourier1}) and (\ref{fourier2}) 
into their real and imaginary part gives the formal identity 
\begin{eqnarray} \label{varterms}
\lefteqn{\int_{\mathbb{R}^{2d}} \ell(x-y) d\mu_2(x,y) - \int_{\mathbb{R}^{2d}} \ell(x-y) d\mu(x)d\mu(y)} \nonumber \\
&=& \int_{\mathbb{R}^d}\hat{\ell}(z)\left[ \int_{P(\mathbb{R}^d)} (Re(\hat{Q}(z)))^2d\nu(Q)-{\left(\int_{P(\mathbb{R}^d)}Re(\hat{Q}(z))d\nu(Q)\right)^2}\right] dz \nonumber \\
&+& \int_{\mathbb{R}^d}\hat{\ell}(z)\left[ \int_{P(\mathbb{R}^d)} (Im(\hat{Q}(z)))^2d\nu(Q)-{\left(\int_{P(\mathbb{R}^d)}Im(\hat{Q}(z))d\nu(Q)\right)^2}\right] dz \nonumber \\
&=& \int_{\mathbb{R}^d}\hat{\ell}(z)\left( \var_{\nu(dQ)}Re(\hat{Q}(z)) + \var_{\nu(dQ)}Im(\hat{Q}(z))\right) dz.
\end{eqnarray}
Here $Re(\hat{Q}(z))$ and $Im(\hat{Q}(z))$ denote the real and the imaginary parts of $\hat{Q}(z)$, and $\var_{\nu(dQ)}Re(\hat{Q}(z))$  and  $\var_{\nu(dQ)}Im(\hat{Q}(z))$ are the variances of the random variables $Re(\hat Q(z))$ and $Im(\hat Q(z))$ with respect to the probability measure $\nu(dQ)$. 

The only steps in the derivation of (\ref{fourier1}), (\ref{fourier2}), (\ref{varterms}) which were  nonrigorous due to lack of regularity of $\ell$ and $Q$ were the use of Plancherel's formula and of the Fourier calculus rule $\widehat{\ell*Q}=\hat{\ell}\hat{Q}$. Conventional assumptions would be 
$\ell*Q$ and $Q\in L^2(\R^d)$ for the former, and $\ell$ and $Q\in L^1(\R^d)$ for the latter. As none of these
four assumptions are actually met here, we will need the following generalization of these facts.  Though this will surely not be surprising to experts in the interest of completeness and for lack of a suitable reference, we include a proof in the Appendix.
\begin{lemma} \label{Fourier} If either $\ell\in C_b(\R^d)\cap L^1(\R^d)$, $\hat{\ell}\ge 0$, or
$\ell$ is the Coulomb cost $\ell(z)=1/|z|$ in dimension $d=3$, and $\int \ell(x-y)\, dQ(x)\, dQ(y)<\infty$, 
$\int \ell(x-y)\, d\tilde{Q}(x)\, d\tilde{Q}(y)<\infty$, then $\hat{\ell}|\hat{Q}|^2$, $\hat{\ell}|\hat{\tilde{Q}}|^2$, $\hat{\ell}\hat{Q}\hat{\tilde{Q}}\in L^1(\R^d)$, and 
\begin{eqnarray} 
\label{quadratic} 
   \int_{\R^{2d}} \ell(x-y)\, dQ(x)\, dQ(y) &=& (2\pi)^{-d} \int_{\R^{d}}\hat{\ell}(z) |\hat{Q}(z)|^2 dz, \\
\label{bilinear}
   \int_{\R^{2d}} \ell(x-y)\, dQ(x)\, d\tilde{Q}(y) &=& (2\pi)^{-d} \int_{\R^{d}}\hat{\ell}(z) \hat{Q}(z) \overline{ \hat{\tilde{Q}}(z) } \,  dz.
\end{eqnarray}
In particular, the identities (\ref{fourier1}), (\ref{fourier2}), (\ref{varterms}) hold. 
\end{lemma}

Now by the assumption $\hat \ell(z) \ge 0$, the two variance terms on the right hand side of (\ref{varterms}) are nonnegative.  
Because the right hand side of (\ref{varterms}) vanishes when $\nu=\delta_\mu$, i.e. $\mu_2=\mu\otimes\mu$, we conclude that $\mu_2=\mu\otimes\mu$ is a minimizer of the problem in (\ref{inf1FOT}), and hence (by Theorem \ref{reduction1}) that $\gamma=\mu^{\otimes\infty}=\mu\otimes\mu\otimes\cdots$ is a minimizer of (\ref{FtildeOTgen}). 
This establishes Theorem \ref{definetti1} (a).

Before proceeding with the proof of (b), let us note a corollary of the above arguments. By combining
(\ref{infFOT}) and (\ref{varterms}), we obtain:
\begin{corollary} \label{repcor} (Probabilistic interpretation of infinite-body optimal transport) 
Let $c(x,y)=\ell(x-y)$ be as in Theorem \ref{definetti1}. If $\gamma\in\calP_{sym}^\infty(\R^d)$, and
if $\nu\in\calP(\calP(\R^d))$ is the unique associated measure from Proposition \ref{definetti} such 
that $\gamma=\int_{\calP(\R^d)} Q^{\otimes\infty} d\nu(Q)$, then the functional $C_\infty$ introduced in
(\ref{FtildeOTgen}) satisfies
$$
  C_\infty[\gamma] = \int_{\R^{2d}} \ell(x-y)\, \mu(dx) \, \mu(dy) + 
                     \int_{\R^d} \hat{\ell}(z) \left( \var_{\nu(dQ)}Re(\hat{Q}(z)) + \var_{\nu(dQ)}Im(\hat{Q}(z)) \right) dz,
$$
where $\mu$ is the one-body marginal of $\gamma$.
\end{corollary}

It remains to show the uniqueness result (b). Suppose $\gamma$ is a minimizer of (\ref{FtildeOTgen}). 
By de Finetti's theorem (\ref{definetti}), there exists a probability measure $\nu\in\calP(\calP(\R^d))$
such that 
\be \label{gammarep}
  \gamma=\int_{\calP(\R^d)} Q^{\otimes\infty} d\nu(Q).
\ee 
We have to show that $\nu$ is the Dirac mass $\delta_\mu$. 
By Theorem \ref{reduction1}, the two-point marginal   $\mu_2=\int_{\calP(\R^d)} Q \otimes Q d\nu(Q)$f $\gamma$ is a minimizer of the problem in (\ref{inf1FOT}). By (\ref{infOTexplic}), and (\ref{varterms}), it follows
that the right hand side of (\ref{varterms}) is zero, i.e. 
\be \label{zero}
  \int_{\calP(\R^d)} |\hat{Q}(z)|^2 d\nu(Q) - \left| \int_{\calP(\R^d)} \hat{Q}(z)\, d\nu(Q)\right|^2 = 0
  \mbox{ for Lebesgue-a.e.}\, z\in\R^d.
\ee
Because the left hand side equals $\int_{\calP(\R^d)} | \hat{Q}(z) - \int_{\calP(\R^d)} \hat{Q}(z)\, d\nu(Q)|^2 d\nu(Q)$, (\ref{zero}) holds if and only if 
$$
  \hat{Q}(z) = \int_{\calP(\R^d)} \hat{Q}(z) \, d\nu(Q) \mbox{ for }\nu-\mbox{a.e.}\, Q\in\calP(\R^d).
$$
Therefore, by the injectivity of the Fourier transform as a map from $\calP(\R^d)$ to $C_b(\R^d)$, 
$$
  Q = \int_{\calP(\R^d)} Q \, d\nu(Q) \mbox{ for }\nu-\mbox{a.e.}\, Q\in\calP(\R^d).
$$
In other words, $\nu$ must be a Dirac mass (at $\mu$, to satisfy the margial constraint). Substitution into
(\ref{gammarep}) shows that $\gamma$ is the independent measure (\ref{indep}). The proof of Theorem \ref{definetti1} is complete. 

\begin{remark}
\begin{itemize}
\item [(a)] The proof of Theorem \ref{definetti1} relies heavily on the positivity of the Fourier transform of $\ell$, and indeed the conclusion can fail dramatically in the absence of this condition, as shown by 
the following example.
\begin{example} \label{Ex:nonposFT} Let $\ell$ be any cost which is zero at $z=0$ and strictly positive elsewhere. Prototypical
are \\[1mm]
(i) the quadratic cost 
$$
                           \ell(z)=|z|^2, 
$$
in which case (\ref{FtildeOTgen}) corresponds to the infinite marginal limit of the problem studied by Gangbo and Swiech in \cite{GS98}, in the special case of equal marginals; physically, one has replaced the repulsive Coulomb interactions by attractive harmonic oscillator-type interactions;
\\[1mm]
(ii) the smoothly truncated quadratic cost 
$$
   \ell(z) = e^{-|z|^2/2\sigma^2} - e^{-\sigma^2|z|^2/2}, \;\sigma>1,
$$
which behaves like $|z|^2$ near $z=0$ (so that (\ref{FOTinfty}) behaves like the quadratic OT problem (i)
for marginals supported near $0$). Note that (ii) satisfies all assumptions of Theorem \ref{definetti1}
except positivity of the Fourier transform $\hat{\ell}$ (note that $\hat{\ell}(k) = 
(\sqrt{2\pi}\sigma)^d e^{-\sigma^2|k|^2/2} - (\sqrt{2\pi}/\sigma)^d e^{-|k|^2/2\sigma^2}$). 
\\[1mm]
It is clear that the probability measure $\gamma:= (Id,Id,...)_{\#}\mu$ (or, in physics notation, 
$\gamma(x_1,x_2,...)=\mu(x_1)\delta_{x_1}(x_2)\delta_{x_1}(x_3)\cdots $) on $(\R^d)^\infty$ satisfies
\begin{equation*}
C_\infty[\gamma] = \int_{\mathbb{R}^{2d}}c(x,y) d\mu_2(x,y) =0, 
\end{equation*}
where $\mu_2$ is the 2-point marginal of $\gamma$. This is because $\mu_2=(Id,Id)_{\#}\mu$ (or, in physics notation, $\mu_2(x,y)=\mu(x)\delta_{x}(y)$) is concentrated on the diagonal $x=y$, where $c(x,y)=|x-y|^2=0$.
Since trivially $C_\infty\ge 0$, the above $\gamma$ is a minimizer. However, by the positivity of $c(x,y)$off the diagonal, the independent measure $\mu\otimes\mu\otimes\cdots$ is not a minimizer except in the trivial case when $\mu=\delta_x$ for some $x\in\R^d$. 
\end{example}

\item [(b)] A representation similar to the Finetti representation (\ref{definettieqn}) but with $\nu\in\calP(\calP(\R^d))$ replaced by a {\it signed} measure has been established in 
\cite{KerSze06}. Such a representation would allow us to derive (\ref{inf1FOT}), but -- due to the lack 
of sign information -- does not allow to conclude that the independent measure is optimal in the finite-$N$
case. Indeed, in the special case of marginals supported on two points it follows from the analysis in 
\cite{CFKMP13} that the independent measure is not minimizing for {\it any} $N$. For more general densities and cost functions, it follows from Proposition \ref{notprod} below that the independent measure is not minimizing for any N.

\item [(c)] Constraints ensuring that the Fourier transform of a function is positive have been derived for example in \cite{GiPe}.

\item [(d)] As a corollary of our analysis, we recover the following interesting result from
from \cite{Hu07}: if $(X_n)_{n\ge 1}$ is an infinite sequence of exchangeable random variables in $\R^d$ such that $(X_n)_{n\ge 1}$ are pairwise independent (i.e., the joint distribution of any $(X_i,X_j)$ is a product of the distributions of $X_i$ and $X_j$), they are mutually independent. Indeed, let $\gamma$ be the
joint distribution of the infinite sequence $(X_1,X_2,...)$, let $\mu_2$ be the distribution of
$(X_1,X_2)$, and let $\mu$ be the distribution of $X_1$. By the assumption of pairwise independence, 
$\mu_2=\mu\otimes\mu$. Hence, fixing for instance the cost $\ell(z)=e^{-|z|^2}$ and combining eq. (\ref{infOTexplic}) and Lemma \ref{reduction}, it follows that $\gamma$ is a minimizer of (\ref{FtildeOTgen}). 
But the uniqueness result of Theorem \ref{definetti1} (b) implies that the only minimizer of (\ref{FtildeOTgen}) is the independent measure $\mu\otimes\mu\otimes\cdots$. Thus $\gamma=\mu\otimes\mu\otimes\cdots$, as was to be shown. 

Note that for $N<\infty$, pairwise independence does not imply mutual independence. One of the first counter-examples for $N<\infty$ was provided in \cite{Bernstein46}; for further counter-examples see e.g. \cite{DerKlop00}.

\item [(e)] We note that weakening even slightly the assumption of exchangeability of the measure may destroy uniqueness of the minimizer of  (\ref{FtildeOTgen}). To prove this, we apply for example the results from \cite{Janson88} or from \cite{Bradley89}. Therein, various examples are constructed of infinite stationary sequences $(X_n)_{n\ge 1}$ of random variables in $\R^d$ such that $(X_n)_{n\ge 1}$ are pairwise independent, with mean $0$ and finite second moments, but which do not satisfy the central limit theorem. This implies that in these particular cases $(X_n)_{n\ge 1}$ are not mutually independent.    
\end{itemize}
\end{remark}

\section{Connection between the N-body OT problem and the infinite-body OT problem}
\label{4}

We will establish in this section the relationship between the infinite-body optimal transport problem (\ref{FtildeOTgen}) and the corresponding 
$N$-body optimal transportation problem (\ref{Nbodyproblem}), as stated in our second main result Theorem \ref{convergence}. We recall first from (\ref{finOT0}) the optimal cost of the $N$-body problem per particle pair, given for all $N\in\bbn$, $N\ge 2$ by
\begin{equation*}
{F}^{OT}_N[\mu] := \frac{1}{{N\choose 2}} \inf_{\gamma\in\calP_{sym}^N(\R^{d}),\gamma\,\mapsto\mu}{C}_N[\gamma].
\end{equation*}

Moreover, analogously to Lemma \ref{reduction} (see also Theorem III.3 in 
\cite{CFKMP13}) we have
\begin{equation}
\label{finOT1}
 {F}^{OT}_{N}[\mu] = \inf\Bigl\{\int_{\R^{2d}} c(x,y) \, d\mu_2(x,y) \; 
   \Big| \; 
   \mu_2\in \calP^2_{sym}(\R^{d}), \, \mu_2\mapsto\mu, \, \mu_2 \mbox{ is N-representable} \Bigr\}.
\end{equation}
This representation will be used in the proof of Theorem \ref{convergence}.

We first note the following existence result for \ref{finOT1}:

\begin{proposition}
\label{existmu2}
Let $c_N\, : , (\R^d)^N \rightarrow {\R}_+\cup\{\infty\}$ be defined as in (\ref{generalcost}), {with $c$ lower semi-continuous}.
Then there exists at least one solution $\gamma_N$ to (\ref{Nbodyproblem}), and at least one solution $\mu_{2,N}\in \calP_{sym}^2(\R^{d})$ to the minimization problem in (\ref{finOT1}).
\end{proposition}
\begin{proof}
The proof follows from a standard compactness argument, similar to those found in \cite{Vill09}, combined with the fact that a non symmetric measure $\gamma$ on $\mathbb{R}^{Nd}$ may be symmetrized without changing the total cost $C_N[\gamma]$, due to the linearity of the functional and the constraints, and the symmetry of $c$.
\end{proof}
To establish (\ref{limitOT}), we will use the following result which allows us to approximate $N$-representable measures by infinitely representable ones.  The result is actually a translation of Theorem $13$ in \cite{DF80} from the language of random variables into that of probability measures. 
For purposes of simplicity and completeness, unlike \cite{DF80} we limit ourselves to euclidean spaces,
and include a proof.

\begin{proposition}
Let $\gamma_N\in\calP_{sym}^N(\R^{d})$. Then there exists an infinitely representable measure ${\mathbb{P}}_{2,\gamma_N}$ such that
\begin{equation}
\label{Finapprox}
||\gamma_2-{\mathbb{P}}_{2,\gamma_N}||\le\frac{1}{N}~~\mbox{and}~~\gamma_1={\mathbb{P}}_{1,\gamma_N}.
\end{equation}
For $1\le k\le N$, we denoted in (\ref{Finapprox}) by $\gamma_k$ the canonical projection of $\gamma_N$ on $P_{sym}(\R^{dk})$ (that is, $\gamma_k\in\calP^k_{sym}(\R^{d})$ is a marginal of $\gamma_N)$, and by $||\gamma_k-\mathbb{P}_{k,\nu}||$ the total variation distance, that is, 
$$||\gamma_k-{\mathbb{P}}_{k,\gamma_N}||:=\sup_{\{f:\R^d\rightarrow\R, \atop f~\mbox{measurable},~ |f|\le 1\}}|\gamma_k(f)-\mathbb{P}_{k,\nu}(f)|.$$ 
\end{proposition}
\begin{proof}
To prove (\ref{Finapprox}), let us define for each $k\ge 1$ the measure ${\mathbb{P}}_{k,\gamma_N}\in {\cal P}(\R^{kd})$ by
\begin{equation}
\label{approxmeasure}
{\mathbb{P}}_{k,\gamma_N}(A_k):=\int_{\R^{Nd}}\left(\frac{{\delta}_{\omega_1}+{\delta}_{\omega_2}+\ldots+{\delta}_{\omega_N}}{N}\right)^{\otimes k}(A_k)\,d\gamma_N(\omega),~~\mbox{for all}~A_k\in\R^{kd}.
\end{equation}

By Kolmogorov's extension theorem, ${\mathbb{P}}_{k,\gamma_N}$ can be extended to an infinite-dimensional symmetric measure ${\mathbb{P}}_{\infty,\nu}$ in $\calP_{sym}^\infty(\R^d)$, which has ${\mathbb{P}}_{k,\gamma_N}$ as marginal for each $k\ge 1$. Moreover, for all $A_2\in {\R}^{2d}$ we obtain from (\ref{approxmeasure}) that
 \begin{eqnarray*}
 {\mathbb{P}}_{2,\gamma_N}(A_2)&=&\int_{\R^{Nd}}\left(\frac{{\delta}_{\omega_1}+{\delta}_{\omega_2}+\ldots+{\delta}_{\omega_N}}{N}\right)^{\otimes 2}(A_2)\,d\gamma_N(\omega)\\
 &=&\frac{N^2-N}{N^2}\gamma_2(A_2)+\frac{1}{N}\gamma_1(\{\omega_1:(\omega_1,\omega_1)\in A_2\}
),
 \end{eqnarray*}
 and therefore
 $$\left |\gamma_2(A_2)- {\mathbb{P}}_{2,\gamma_N}(A_2) \right |=\frac{1}{N}\left | \gamma_2(A_2)-\gamma_1(\omega_1:(\omega_1,\omega_1)\in A_2)\right |\le\frac{1}{N}.
$$
\end{proof}
We will use this result directly to easily establish Theorem \ref{convergence} part (i).  For part (ii), we will need the following intermediate Lemma.
\begin{lemma}\label{infrepNrep}
A symmetric measure $\mu_2$ on $\mathbb{R}^{2d}$ is infinitely representable if and only if it is $N$-representable for all $N$.
\end{lemma}
\begin{proof}
It is clear that an infinitely representable measure is $N$-representable for all $N$.  On the other hand, if $\mu_2$ is $N$-representable for all $N$, the preceding result yields a sequence of infinitely representable measures converging weakly to $\mu_2$.  By the weak closedness of the set $\calP_{sym}^\infty(\R^d)$ (see \cite{Ald85}), the set of representable symmetric measures on $\mathbb{R}^{2d}$ is weakly closed and the result follows.
\end{proof}

{\bf Proof of Theorem \ref{convergence}}
We first prove part (i) (the bounded cost case) directly from Proposition \ref{Finapprox}.  Letting $\mu_{2,N}$ solve (\ref{finOT1}), we have by Proposition \ref{Finapprox} an infinitely representable $\mu_{2,\infty}$ with 1-body marginal $\mu$ such that $||\mu_{2,N}-\mu_{2,\infty}|| \leq \frac{1}{N}$.  Therefore
\begin{eqnarray}
F_N^{OT}[\mu] &=& \int_{\mathbb{R}^{2d}}c(x,y)d\mu_{2,N}\\
&\geq &\int_{\mathbb{R}^{2d}}c(x,y)d\mu_{2,\infty} -\frac{||c||_\infty }{N}\\
&\geq & F_{\infty}^{OT}[\mu] -\frac{||c||_\infty }{N}. \label{similartolieboxford}
\end{eqnarray}
Noting that $F_N^{OT}[\mu] \leq F_{\infty}^{OT}[\mu]$ and taking the limit in the above inequality yields the result. 

To prove assertion (i), we use Lemma \ref{infrepNrep}.  Let $\mu_{2,N}$ solve (\ref{finOT1}).  By the tightness of the set of symmetric measures on $\mathbb{R}^{2d}$ with common marginal $\mu$ and by Prokhorov's theorem, we can, after passing to a subsequence, assume $\mu_{2,N}$ converges to some symmetric  $\mu_2$ whose marginal is also $\mu$. Now, it is clear that the set of $M$-representable measures is weakly closed (for a proof of this statement, see e.g. page 54 in \cite{Ald85}), and therefore, $\mu_2$ is $M$-representable for each fixed $M$ (as $\mu_{2,N}$ is $M$-representable for $N \geq M$, by Lemma \ref{Nrepressmall}).  Therefore, by the preceding Lemma, $\mu_2$ is infinitely representable.

By lower semi-continuity of $c$, we therefore have 

\begin{equation}
\liminf_{N\rightarrow \infty}F_N^{OT}[\mu] =\liminf_{N\rightarrow \infty} \int_{\mathbb{R}^{2d}}c(x,y)d\mu_{2,N} \geq \int_{\mathbb{R}^{2d}}c(x,y)d\mu_{2} \geq F_{\infty}^{OT}[\mu].
\end{equation}

As we clearly have $F_N^{OT}[\mu] \leq F_{\infty}^{OT}[\mu]$ for each $N$, this implies the desired result.

\qed

\begin{remark}
\begin{itemize}
\item [(a)] We note here that the proof in fact yields that any convergent subsequence of optimal $\mu_{2,N}$ in the $N$-body problem converges to a solution to the infinite body problem.  Whenever the minimizer $\mu_{2,\infty}$ in the infinite body problem is unique (for example, under the conditions in Theorem \ref{definetti1} part (ii)), this implies that the $\mu_{2,N}$ converge to $\mu_{2,\infty}$.  For bounded costs, the proof also yields abound on the rate of convergence of $\frac{||c||_\infty}{N}$.

\item [(b)] Theorem $13$ from \cite{DF80} proves the following: Let $\gamma_N\in\calP_{sym}(\R^{dN})$. Then there exists a measure $\nu$ on the set of probability measures on ${\cal P}(\R^d)$, such that
\begin{equation}
\label{Finapprox1}
||\gamma_k-{\mathbb{P}}_{k,\nu}||\le\frac{k(k-1)}{N}~\mbox{for all}~1\le k\le N.
\end{equation}
For some particular cases of marginals $\gamma_1$ the bounds in (\ref{Finapprox1}) have been improved in \cite{DF87}.
\end{itemize}
\end{remark}

Next we point out a variant of our result in Corollary \ref{IHDL} on the inhomogeneous high-density limit of the SCE functional introduced in (\ref{FOTN}), (\ref{CN}). By eq. (\ref{Veescaled}) together
with the characterization (\ref{finOT1}) of $F^{OT}_N$ as an infimum over representable pair measures
(or alternatively Theorem III.3 in \cite{CFKMP13}), we have
\begin{equation}
\label{finOT01}
{V}_{ee}^{SCE}[\rho] = {N\choose 2} \inf\Bigl\{\int_{\R^{6}} \frac{1}{|x-y|} \, d\mu_2(x,y) \; 
   \Big| \; 
   \mu_2\in \calP^2_{sym}(\R^{3}), \, \mu_2\mapsto\rho/N, \, \mu_2 \mbox{ N-representable} \Bigr\},
\end{equation}
where $\rho$ is any integrable nonnegative function on $\R^3$ with $\int_{\R^3}\rho=N$. 
This formula suggests a natural hierarchy
of approximations as introduced in \cite{CFKMP13}: for $k=2,3,...$ we define
\begin{equation}
\label{restricfinOT}
{V}_{ee}^{SCE,k}[\rho] :={N\choose 2} \inf\Bigl\{\int_{\R^{6}} \frac{1}{|x-y|} \, d\mu_2(x,y) \; 
   \Big| \; 
   \mu_2\in \calP_{sym}(\R^{6}), \, \mu_2\mapsto\rho/N, \, \mu_2 \mbox{ k-representable} \Bigr\}.
\end{equation}
That is, we replace the requirement that $\mu_2$ is $N$-
representable by the modified requirement that it be $k$-representable. Because $k$-representability becomes a stronger and stronger condition as $k$ increases, we have the following
chain of inequalities
$${V}_{ee}^{SCE,2}[\rho]\le\cdots\le  V_{ee}^{SCE,3}[\rho]\le \cdots \le {V}_{ee}^{SCE,N}[\rho] = V_{ee}^{SCE}[\rho]\le{V}_{ee}^{SCE,N+1}[\rho]\le \cdots \,.$$
The functionals ${V}_{ee}^{SCE,k}$  
can be thought of as reduced models for the energy of
strongly correlated electrons which take into account $k$-
body correlations. 
\begin{corollary}
\label{selfinter} Assume that $\rho\in L^1(\R^3)$, $\rho\ge 0$, $\int_{\R^3}\rho = N$ for some natural number $N\ge 2$. Then 
\begin{equation}
\label{vscelimit}
\lim_{k\rightarrow\infty} V_{ee}^{SCE,k}[\rho] = \frac12 (1-\frac{1}{N}) \int_{R^6}\frac{1}{|x-y|}\rho(x)\rho(y) dx dy.
\end{equation}
\end{corollary}

Physically, the factor $1-1/N$ is a self-interaction correction, and the right hand side of (\ref{vscelimit}) is a self-interaction corrected mean field energy. Thus the approximation via density representability of infinite order remembers that there are only ${N\choose 2}$ interaction
terms, not $N^2/2$. 

\begin{proof} By the definition (\ref{restricfinOT}), for any $\rho$ as above we have
$$
   V_{ee}^{SCE,k}[\rho] = {N\choose 2} F_k^{OT}[\rho/N],
$$
that is to say, up to scaling factors $V_{ee}^{SCE,k}[\rho]$ is the optimal cost of a $k$-body optimal 
transport problem. By Theorems \ref{convergence} and \ref{definetti1}, the right hand side converges to
$$
   {N\choose 2} \int_{\R^6} \frac{ \frac{\rho(x)}{N} \, \frac{\rho(y)}{N} }{|x-y|} \, dx \, dy
$$
as $k\to\infty$. This establishes the corollary.  
\end{proof}

Finally we note that, in contrast to the $N=\infty$ case, minimizers of the $N$-body optimal transport problem exist are typically not given by the mean field measure for any $N<\infty$. 

\begin{proposition}
\label{notprod}
Let $c_N:(\R^d)^N \rightarrow {\R}_+\cup\{\infty\}$ be defined as in (\ref{generalcost}). Assume that there is some point $x=(x_1,x_2,...,x_N) \in \mathbb{R}^{Nd}$ such that $c_N$ is $C^2$ near 
$x$, $D^2_{x_ix_j}c(x) \neq 0$ for some $i \neq j$, and the measure $\mu$ has positive density near each $x_i \in \mathbb{R}^d$.  Then the product measure $\mu\otimes\mu$ on $\mathbb{R}^{2d}$ is not optimal for the $2$-body optimal transport problem with $N$-representability constraint  (\ref{finOT1}), for any $N < \infty$.

Note that for the Coulomb cost, the conditions on the cost hold for any $x=(x_1,x_2,...,x_N)$ away from the diagonal; that is, for any $x$ such that $x_i \neq x_j$ for all $i \neq j$.
\end{proposition}
\begin{proof}
Fix $N<\infty$. The proof is by contradiction; assume that the product measure  $\mu^{\otimes 2}$ on $\mathbb{R}^{2d}$ is optimal for (\ref{finOT1}). Then the product measure $\mu^{\otimes N}$ on $\mathbb{R}^d \times \mathbb{R}^d \times ...\times \mathbb{R}^d$ must be optimal for the $N$-body optimal transport formulation of the problem (\ref{FtildeOTgen}).  It is clear that the support of the product measure has full Hausdorff dimension $dN$ near the point $x$. On the other hand, Theorem 2.3 from \cite{Pass12} implies that for any optimizer $\gamma$, for some neighbourhood $U$ of $x$, the dimension of the $supp(\gamma) \cap U$ is no more than  $\lambda_0+\lambda_-$, where $supp(\gamma)$ is the support of $\gamma$, and $\lambda_+, \lambda_-$ and $\lambda_0$ are respectively the number of positive, negative and zero eigenvalues of the off-diagonal part of the Hessian

\begin{equation} \qquad
G=\label{G}
\begin{bmatrix}
0 & D^{2}_{x_1x_2}c & D^{2}_{x_1x_3}c & ...&D^{2}_{x_1x_N}c \\ 
D^{2}_{x_2x_1}c & 0 & D^{2}_{x_2x_3}c & ...&D^{2}_{x_2x_N}c \\ 
D^{2}_{x_3x_1}c & D^{2}_{x_3x_2}c & 0 & ...&D^{2}_{x_3x_N}c \\
...&...&...&...&...,\\
D^{2}_{x_Nx_1}c & D^{2}_{x_Nx_2}c & D^{2}_{x_Nx_3}c & ...&0
\end{bmatrix}
\end{equation}
evaluated at $x$. Therefore, if $\mu^{\otimes N}$ is optimal, $G$ must have no positive eigenvalues and therefore must be negative semi-definite.  This is clearly not true; as $D^{2}_{x_ix_j}c \neq 0$, we can choose $u, v \in \mathbb{R}^d$ such that $u \cdot D^{2}_{x_ix_j}c\cdot v^T >0$.  Then 
\begin{eqnarray*}
[0,...,0,u,0,...,0,v,0,....,0] \cdot G \cdot [0,...,0,u,0...,0,v, ,0,...,0]^T &= &v \cdot D^{2}_{x_jx_i}c\cdot u^T + u\cdot D^{2}_{x_ix_j}c\cdot v^T \\
&=& 2u\cdot D^{2}_{x_ix_j}c\cdot v^T\\
& >& 0,
\end{eqnarray*}
contradicting the negative definiteness of $G$.

\end{proof}

\section{Appendix}
Here we prove the result stated in Lemma \ref{Fourier} that a well known formula from Fourier transform calculus on $\R^d$ remains valid for integrals involving two probability measures and a cost function such as the Coulomb cost. 

The formula would be straightforward if the probability measures and the cost function belonged to $L^1(\R^d)$. The generalization to arbitrary probability measures was essential in the proof of our main result that 
the solution to infinite-body optimal transport problems for costs with positive Fourier transform is the independent product measure. We note that the generalization is needed even in the case of smooth marginals, since general probability measures always appear in the de Finetti 
representation (\ref{definettieqn}) of trial measures.

{\bf Proof of Lemma \ref{Fourier}}. First we deal with the case $\ell\in C_b(\R^d)\cap L^1(\R^d)$. 
We begin by proving (\ref{quadratic}). The idea is to regularize $Q$. Let $G_\eps$ be a Gaussian with standard deviation $\eps$, i.e. $G_\eps(x)=(2\pi\eps^2)^{-d/2} e^{-|x|^2/2\sigma^2}$. Then 
$\widehat{G_\eps}(k) = e^{-\eps |k|^2/2}$. By inspection, $\hat{G_\eps}$ converges {\it monotonically}
to $1$ as $\eps\to 0$. The monotonicity of this convergence is actually needed in the argument below.

Now for any given probability measure $Q$ on $\R^d$, let $Q_\eps$ be the regularization $Q_\eps(x)=(G_\eps * Q)(x) = \int_{\R^d} G_\eps(x-y)\, dQ(y)$. Then $Q_\eps\in L^1(\R^d)\cap L^\infty(\R^d)$; in particular $Q_\eps\in L^2$. Next we claim that $\ell*Q_\eps\in L^2(\R^d)$. This is because $\ell*Q_\eps$ is, as a convolution of two $L^1$ functions, in $L^1$, and also, as a convolution of an $L^1$ and an $L^\infty$ function, in $L^\infty$. 

Since $\ell$ and $Q_\eps$ are in $L^1(\R^d)$, it is straightforward from the definition of the Fourier transform on $L^1$ as a convergent integral that $\widehat{\ell * Q_\eps} = \widehat{\ell} \, \widehat{Q_\eps}$. It follows that formula (\ref{quadratic}) is valid for the regularized measure $Q_\eps$, i.e.
\be \label{approxquad}
  \int_{\R^{2d}} \ell(x-y)\, dQ_\eps(x) dQ_\eps(y) = (2\pi)^{-d} \int_{\R^d} \hat{\ell}(z) \, |\hat{Q_\eps}(z)|^2 dz.
\ee
It remains to pass to the limit $\eps\to 0$. Since $Q_\eps\rightharpoonup Q$ weakly (that is to say 
$\int_{\R^d} \varphi Q_\eps \to \int_{\R^d} \varphi \, dQ$ for all $\varphi$ belonging to the 
space $C_b(\R^d)$ of bounded continuous functions), we have $Q_\eps\otimes Q_\eps \rightharpoonup
Q\otimes Q$, and since the function $(x,y)\mapsto \ell(x,y) \in C_b(\R^{2d})$ we infer that the left
hand side of (\ref{approxquad}) converges to the left hand side of (\ref{quadratic}). Since
$\widehat{Q_\eps} = \widehat{G_\eps} \, \widehat{Q}$, $\widehat{\ell}\ge 0$, and $\widehat{G_\eps}$
converges monotonically to $1$, the integrand on the right hand side of (\ref{approxquad}),
$\widehat{\ell} |\widehat{Q_\eps}|^2 = \hat{\ell}|\widehat{G_\eps}|^2|\widehat{Q}|^2$, converges
monotonically to $\widehat{\ell}$. Hence by monotone convergence, the right hand side of (\ref{approxquad}) tends to that of (\ref{quadratic}), establishing (\ref{quadratic}). 

It remains to prove (\ref{bilinear}). Analogously to the proof of (\ref{quadratic}) we obtain
\be \label{approx}
  \int_{\R^{2d}} \ell(x-y) \, Q_\eps(x) \, \tilde{Q}_\eps(y) \, dx \, dy = 
  (2\pi)^{-d} \int_{\R^d} \widehat{\ell} \widehat{Q_\eps} \overline{\widehat{Q_\eps}}
\ee
as well as the fact that the left hand side tends to the left hand side of (\ref{bilinear})
as $\eps\to 0$. The argument for passing to the limit on the right hand side no longer works, 
since now the integrand is not in general nonnegative. Instead we use that by the assumption of
finiteness of $\int\ell(x-y)dQ(x)dQ(y)$ and $\int\ell(x-y)d\tilde{Q}(x)d\tilde{Q}(y)$ and by
(\ref{quadratic}), $\hat{\ell}|\widehat{Q}|^2$ and $\hat{\ell}|\widehat{\tilde{Q}}|^2$ are in 
$L^1(\R^d)$. This together with the pointwise estimate
$$
    | \widehat{Q_\eps} \, \widehat{\tilde{Q}_\eps} | \le \frac12 \left(    |\widehat{Q}|^2 + 
    |\widehat{\tilde{Q}}|^2 \right)
$$
(which relies on $\widehat{Q_\eps} = \widehat{G_\eps}\widehat{Q}$ and $|\widehat{G_\eps}|\le 1$)
shows that the convergence $\widehat{\ell}\widehat{Q_\eps}\widehat{\tilde{Q}_\eps}\to
\widehat{\ell}\widehat{Q}\widehat{\tilde{Q}}$ is dominated. Hence by the dominated convergence theorem the right hand side of (\ref{approx}) tends to that of (\ref{bilinear}) as $\eps\to 0$. 
This completes the proof of Lemma \ref{Fourier} in the case $\ell\in C_b\cap L^1$. 

It remains to deal with the Coulomb case $d=3$, $\ell(x)=\frac{1}{|x|}$. In this case the above proof does not work, for instance because weak convergence of the probability measure $Q_\eps\otimes Q_\eps$ is insufficient to pass to the limit in the left hand side of (\ref{approxquad}) due to the fact that $(x,y)\mapsto\ell(x,y)$ no longer belongs to the space $C_b$ of bounded continuous functions associated by duality. However the desired Fourier identities were established 
in \cite{CF09}, with passage to the limit in (\ref{approxquad}) being achieved with the help of Newton's screening theorem. The latter is the special Coulombic property
that for any continuous radially symmetric function $\varphi$ with compact support, $\varphi*1/|\cdot|=1/|\cdot|$ outside the support of $\varphi$ (or, physically speaking, the 
potential exerted by a radial charge distribution onto a point outside it is the same as that of the point charge obtained by placing all its mass at the center). 
\\[2mm]

\section{Conclusions}
Mean field approximations that reduce complicated many-body interactions to interactions of each particle with a collective mean field
are ubiquitous in many areas of physics such as quantum mechanics, statistical mechanics, electromagnetism, 
and continuum mechanics, as well as in other fields such as mathematical biology, probability theory, 
or game theory. 

Motivated by questions in many-electron quantum mechanics, we have presented a novel and quite general 
mathematical picture of how mean field approximations are rigorously related to underlying many-body 
interactions. Namely, for interactions with positive Fourier transform
they emerge as the unique solution to a naturally associated infinite-body optimal 
transport problem.

\section{Acknowledgements}
This work was begun when all three authors attended the 2012 trimester program at the Hausdorff Research Insitute for Mathematics in Bonn on 'Mathematical
challenges of materials science and condensed matter physics'. We wish to thank the program organizers Sergio Conti, Richard James, Stephan Luckhaus, Stefan M\"uller, Manfred Salmhofer, and Benjamin Schlein for their hospitality. We are also indebted to Paola Gori-Giorgi for pointing out to us a very interesting alternative proof of Corollary 1.3  which is implicit in Ref. \cite{RSG11} (see Remark 1.4). B.P.  acknowledges the support of a University of Alberta start-up grant and National Sciences and Engineering Research Council of Canada Discovery Grant number 412779-2012.


\begin{thebibliography}{99}

\bibitem[AD06]{AD06}{\sc P. W. Ayers, E. R. Davidson.} Necessary conditions for the N-representability of Pair Distribution Functions. {\it Int. J. Quantum Chemistry} 106, 1487-1498 (2006).


\bibitem[Ald85]{Ald85}{\sc D. Aldous.} Exchangeability and Related Topics. {\it Ecole d'Ete St Flour 1983}. Springer Lecture Notes in Math. {\bf 1117}, 1-198 (1985).


\bibitem[Be93]{Be93} {\sc A. Becke.} Density-functional thermochemistry. III. The role of exact exchange. {\it J. Chem. Phys.} {\bf 98}, 5648 (1993).

\bibitem[Bernstein46]{Bernstein46} {\sc S.N. Bernstein.} Theory of probability. {\it Gostechizat.} Moscow-Leningrad, 4th edition (in Russian) (1946).

\bibitem[Bill99]{Bill99}{\sc P. Billingsley.} Convergence of Probability Measures. New York. {\it John Wiley \& Sons.} (1999).

\bibitem[BorLew06]{BorLew06}{\sc J. Borwein, A. Lewis.} Convex Analysis and Nonlinear Optimization: Theory and Examples (2 ed.), {\it Springer} (2006).

\bibitem[Bradley89]{Bradley89}{\sc R. Bradley.} A stationary, pairwise independent, absolutely regular sequence for which the central limit theorem fails. {\it Prob. Theory Rel. Fields}, {\bf 81}, 1, 1-10 (1989). 

\bibitem[Bre87]{Bre87}{\sc Y. Brenier.} Decomposition polaire et rearrangement monotone des champs
de vecteurs. {\it C.R. Acad. Sci. Pair. Ser. I Math.}, {\bf 305}, 805�808, (1987)

\bibitem[BPG12]{BPG12}{\sc G. Buttazzo, L. De Pascale, P. Gori-Giorgi.} Optimal transport formulation of electronic density-functional theory. {\it Phys. Rev.
A}, {\bf 85}, 062502 (2012).
{\bibitem[CF09]{CF09}{\sc S. Capet, G. Friesecke.} Minimum energy configurations of classical charges: Large N asymptotics. {\it Appl. Math. Research Express.}, doi:10.1093/amrx/abp002, (2009).}

\bibitem[Car03]{Car03} {\sc G. Carlier.} On a class of multidimensional optimal transportation problems. {\it J. Convex Anal.} {\bf 1}, 517-529 (2003).

\bibitem[CarNaz08]{CarNaz08} {\sc G. Carlier, B. Nazaret.} Optimal transportation for the determinant. {\it ESAIM Control Optim. Calc. Var.} {\bf 4}, 678-698 (2008).

\bibitem[COLYUK00]{COLYUK00}{\sc A.J. Coleman, V.I. Yukalov.} Reduced Density Matrices: Coulson's Challenge. Lecture Notes in Chemistry. {\it Springer.} (2000).

\bibitem[CD13]{CD13}{\sc M. Colombo, S. Di Marino.}  Equality between Monge and Kantorovich multimarginal problems with Coulomb cost. \textit{Preprint} (2013).


\bibitem[CDD13]{CDD13}{\sc M. Colombo, L. De Pascale, S. Di Marino}  Multimarginal optimal transport maps for 1-dimensional repulsive costs. \textit{Preprint} (2013).

\bibitem[CFK11]{CFK11} {\sc C. Cotar, G. Friesecke, C. Kl\"uppelberg.} Density functional theory and optimal transportation with  Coulomb cost, {\it Comm. Pure. Appl. Math.}, {\bf 66}, 4, 548-599 (2013).
\bibitem[CFK12]{CFK2}{\sc C. Cotar, G. Friesecke, C. Kl\"uppelberg.} Smoothing
of transport plans with fixed marginals and rigorous
semiclassical limit of the Hohenberg-Kohn functional, \textit{preprint.} (2013).

\bibitem[Da95]{Da95}{\sc E. R. Davidson.} N-representability of the electron pair density. {\it Chemical Physics Letters} 246, 209-213 (1995). 

\bibitem[deFin69]{deFin69}{\sc B. de Finetti.} Sulla proseguibilit\'a di processi aleatori scambiabili, Rend. Matem.
Trieste, {\bf 1}, 53-67 (1969).

\bibitem[DerKlop00]{DerKlop00}{\sc Y. Derriennic, A. Klopotowski.} On Bernstein's example of three pairwise independent random variables, {\it Sankhya: The Indian Journal of Statistics.}, {\bf 62}, A, 3, 318-330  (2000).

\bibitem[DF80]{DF80}{\sc  P. Diaconis, D. Freedman.} Finite Exchangeable Sequences, {\it Ann. Probab.}, {\bf 8}, 4, 745-764 (1980).


\bibitem[DF87]{DF87}{\sc P. Diaconis, D. Freedman.} A dozen de Finetti-style results in search of a theory. {\it Ann. Inst. H. Poincar� Probab. Statist.}, {\bf 23},  2, 397�423 (1987).
\bibitem[DunSch58]{DunSch58} {\sc N. Dunford, J.T. Schwartz.} Linear operators, Part I, {\it Wiley-Interscience}, (1958).




\bibitem[FNM03]{FNM03} {\sc C. Fiolhais, F. Noqueira, M. Marques (eds).} A Primer in Density Functional Theory. {\it Springer Lecture Notes in Physics} Vol. 620 (2003)

\bibitem[Fr03]{Fr03} {\sc G. Friesecke.} The multiconfiguration equations for atoms and molecules: charge quantization and existence of solutions. {\it Arch. Rat. Mech. Analysis} {\bf 169}, 35-71 (2003)


\bibitem[CFKMP13]{CFKMP13}{\sc G. Friesecke, C. Mendl, B. Pass, C. Cotar, C. Kl\"uppelberg.}  N-density representability and the optimal transport limit of the Hohenberg-Kohn functional, {\it Journal of Chemical Physics}, {\bf 139} (2013).

\bibitem[FRIGRA97]{FRIGRA97} {\sc  B. E.  Fristedt, L. F. Gray.} A Modern Approach to Probability Theory, {\it Birkh\"auser}, Boston (1997).

\bibitem[GalGhou13]{GalGhou13}{\sc A. Galichon, N. Ghoussoub.} Variational representations for N-cyclically monotone vector fields.  To appear in {\it Pacific Journal of Mathematics}.

\bibitem[GM95]{GM95} {\sc W. Gangbo, R. McCann.} Optimal maps in Monge's mass transport problem, {\it C.R. Acad. Sci. Paris. Ser. I. Math.} {\bf 325}, 1653-1658 (1995)

\bibitem[GM96]{GM96} {\sc W. Gangbo, R. McCann.} The geometry of optimal transportation, {\it Acta Math.} {\bf 177}, 113-161 (1996).

\bibitem[GS98]{GS98} {\sc W. Gangbo, A. Swiech.} Optimal Maps for the Multidimensional Monge-Kantorovich Problem. {\it Comm. Pure Applied Math.} {\bf 1}, 23-45 (1998).


\bibitem[GhouMaur13]{GhouMaur13}{N. Ghoussoub, B. Maurey.} Remarks on multi-marginal symmetric Monge-Kantorovich problems. To appear in {\it Discrete and Continuous Dynamical Systems-A.} (2013).

\bibitem[GhouMoa13a]{GhouMoa13a}{\sc N. Ghoussoub, A. Moameni.} A Self-dual Polar Factorization for Vector Fields.  {\it Comm. Pure. Applied. Math.} {\bf66},  905-933 (2013).

\bibitem[GhouMoa13]{GhouMoa13}{N. Ghoussoub, A. Moameni.} Symmetric Monge-Kantorovich problems and polar decompositions
of vector fields, {\it preprint.} (2013).




\bibitem[GiPe]{GiPe}{\sc B.P. Giraud, R. Peschanski.} On positive functions with positive Fourier transforms. {\it Acta Physica Polonica B.}, {\bf 2}, 37, 331-346 (2006).


\bibitem[HK64]{HK64} {\sc P. Hohenberg, W. Kohn.} Inhomogeneous electron gas, {\it Phys. Rev. B} {\bf 136}, 864-871 (1964).
\bibitem[Hein02]{Hein02}{\sc H. Heinich.} Probleme de Monge pour n probabilities. {\it C.R. Math. Acad. Sci. Paris,
334(9)}, 793-795 (2002).

\bibitem[Hu07]{Hu07}{\sc T.C. Hu.} On pairyise independent and independent exchangeable random variables, {\it Stochastic Analysis and Applications.}, {\bf 15}, 1, 51-57 (2007).

\bibitem[Ito86]{Ito86}{\sc K. Ito.} An introduction to Probability Theory. {\it Cambridge University Press} (1986).

\bibitem[KimPass13]{KimPass13}{\sc Y.-H. Kim, B. Pass.}  Multi-marginal optimal transport on Riemannian manifolds. \textit{Preprint} (2013).

\bibitem[KATZ]{KATZ} {\sc Y. Katznelson.} An introduction to harmonic Analysis. {\it Cambridge University Press} (2004).


 \bibitem[KerSze06]{KerSze06}{\sc J.G. Kerns, G.J. Szekely.} De Finetti�s theorem for abstract finite exchangeable sequences. {\it Journal of Theoretical Probability}, {\bf 19}, 3, 589�608 (2006).


\bibitem[KS65]{KS65} {\sc W. Kohn,  L. J. Sham.} Self-consistent equations including exchange and correlation effects. {\it Phys. Rev. A} {\bf 140}, 1133-1138 (1965)

\bibitem[Janson88]{Janson88}{\sc S. Janson.} Some pairwise independent sequences for which the central limit theorem fails. {\it Stochastics}, {\bf 23}, 4, 439-448 (1988). 

\bibitem[Le79]{Le79} {\sc M. Levy.} Universal variational functionals of electron densities, first-order density matrices, and natural spin-orbitals and solution of the v-representability problem. {\it Proc. Natl. Acad. Sci. USA} {\bf 76}(12), 6062-6065 (1979)

\bibitem[Li83]{Li83}{\sc E.H. Lieb.} Density functionals for Coulomb systems, {\it International Journal of Quantum Chemistry} {\bf 24}, 243-277 (1983)

\bibitem[LO81]{LO81} {\sc E.H. Lieb, S. Oxford.} Improved lower bound on the indirect Coulomb energy, {\it International Journal of Quantum Chemistry} {\bf 19}, 427–439 (1981).



\bibitem[Pass11]{Pass11}{\sc B. Pass.} Uniqueness and Monge solutions in the multi-marginal optimal
transportation problem. {\sc SIAM J. Math. Anal.} {\bf 43}, 2758-2775 (2011).
\bibitem[Pass10]{Pass10}{\sc B. Pass.} On the local structure of optimal measures in the
multi-marginal optimal transportation problem. {\it Calc. Var. and PDE.} {\bf 43}, 529-536 (2012).
\bibitem[Pass12]{Pass12}{\sc B. Pass.} An upper bound on the semi-classical Hohenberg-Kohn functional. {\it Nonlinearity} \textbf{26}, 2731-2744 (2013).
\bibitem[Pass12a]{Pass12a}{\sc B. Pass.} Optimal transportation with infinitely many marginals. {\it J. Funct. Anal.} {\bf 264}, 947-963 (2013).
\bibitem[Pass12b]{Pass12b}{\sc B. Pass.}  On a class of optimal transportation problems with infinitely many marginals. \textit{SIAM J. Math. Anal.} \textbf{45}, 2557-2575 (2013).
\bibitem[Pass12c]{Pass12c}{\sc B. Pass.}  Multi-marginal optimal transport and multi-agent matching problems: uniqueness and structure of solutions. \textit{Discrete Contin. Dyn. Syst.} \textbf{34}, 1623-169 (2014).
\bibitem[PY95]{PY95} {\sc R. G. Parr, W. Yang.} {\it Density-Functional Theory of Atoms and Molecules.} Oxford University Press, Oxford (1995)

\bibitem[Ra09]{Ra09} {\sc D. Rappoport, N. R. M. Crawford, F. Furche, K. Burke.} Which density functional should I choose? In: {\it Computational Inorganic and Bioinorganic Chemistry}, eds. E. I. Solomon, R. B. King, and R. A. Scott. Wiley (2009)

\bibitem[RSG11]{RSG11} {\sc E. R\"as\"anen, M. Seidl, P. Gori-Giorgi.} Strictly correlated uniform electron droplets. {\it Phys. Rev. B} {\bf 83}, 195111 (2011).


\bibitem[Rud87]{Rud87}{\sc W. Rudin.} Real and Complex Analysis (Third ed.) {\it Singapore: McGraw Hill} (1987).


\bibitem[Rus91]{Rus91}{\sc L. R\"uschendorf.} Bounds for distributions with multivariate marginals. In: {\it Stochastic Orders and Decisions, Eds.: K. Mosler, M. Scarsini, IMS Lecture Notes 19}, 285-310 (1991).

\bibitem[RusUck97]{RusUck97}{\sc L. R\"uschendorf, L. Uckelmann.} On optimal multivariate couplings. In: {\it Proceedings of Prague 1996 conference on marginal problems, Eds.: V. Benes, I. Stepan. Kluwer}, 261-274 (1997).

\bibitem[Seidl99]{Seidl99}
{\sc Michael Seidl.}
\newblock {Strong-interaction limit of density-functional theory}.
\newblock {\em Phys. Rev. A} 60:4387--4395 (1999).

\bibitem[SGS07]{SGS07}{\sc M. Seidl, P. Gori-Giorgi, A. Savin.} Strictly correlated
electrons in density-functional theory: A general formulation with applications to spherical densities. {\it Phys. Rev. A}, {\bf 75}, 042511 (2007).


\bibitem[SPL99]{SPL99}
{\sc Michael Seidl, John~P. Perdew, and Mel Levy.}
\newblock {Strictly correlated electrons in density-functional theory}.
\newblock {\em Phys. Rev. A} 59:51--54 (1999).




\bibitem[Vill09]{Vill09}{\sc C. Villani} Optimal Transport: Old and New {\it Springer, Heidelberg} (2009).

\end{thebibliography}
\end{document}